\documentclass[a4paper,11pt]{article}
\usepackage{mathrsfs,amssymb,url,color}
\usepackage{amsthm}
\usepackage{amsmath}
\usepackage{hyperref}
\hypersetup{colorlinks=true,linkcolor=black,citecolor=magenta}
\usepackage[margin=1.2in]{geometry}
\usepackage{nicefrac}
\usepackage{enumerate}
\usepackage{xspace}
\usepackage{mathpazo}

\def\bs{\mathbf{s}}
\def\bc{\mathbf{c}}
\def\be{\mathbf{e}}
\def\bv{\mathbf{v}}
\def\ba{\mathbf{a}}
\def\bb{\mathbf{b}}
\def\bx{\mathbf{x}}
\def\bt{\mathbf{t}}
\def\by{\mathbf{y}}
\def\bs{\mathbf{s}}
\def\bzero{\mathbf{0}}
\def\bone{\mathbf{1}}
\def\E{\mathbb{E}}
\def\tA{\tilde{A}}
\def\tB{\tilde{B}}
\newcommand{\rs}{\mathsf{RS}}
\def\ed{\mathsf{ed}}

\newtheorem{theorem}{Theorem}[section]

\newtheorem{lemma}[theorem]{Lemma}

\newtheorem{claim}[theorem]{Claim}
\newtheorem{proposition}[theorem]{Proposition}
\newtheorem{definition}[theorem]{Definition}

\newtheorem{open}[theorem]{Open Question}

\newcommand{\Wone}{{$\mathsf{W[1]}$}\xspace}
\newcommand{\GapETH}{{$\mathsf{Gap}$-$\mathsf{ETH}$}\xspace}
\newcommand{\ETH}{{$\mathsf{ETH}$}\xspace}
\newcommand{\SETH}{{$\mathsf{SETH}$}\xspace}
\newcommand{\OVH}{{$\mathsf{OVH}$}\xspace}
\newcommand{\OV}{{$\mathsf{OV}$}\xspace}
\newcommand{\MIP}{{$\mathsf{MIP}$}\xspace}
\newcommand{\BMIP}{{$\mathsf{BMIP}$}\xspace}
\newcommand{\IP}{{$\mathsf{IP}$}\xspace}
\newcommand{\ABMIP}{{$\mathsf{Additive\text{-}BMIP}$}\xspace}

\def\N{{\mathbb {N}}}

\def\N{{\mathbb {N}}}

\def\poly{{\rm {poly}}}

\def\eps{\varepsilon}
\def\cC{\mathcal{C}}

\def\cB{\mathcal{B}}
\def\cA{\mathcal{A}}
\def\cP{\mathcal{P}}
\def\cQ{\mathcal{Q}}

\newcommand{\CP}{$\mathsf{CP}$\xspace}
\newcommand{\BCP}{$\mathsf{BCP}$\xspace}

\newcommand{\cd}{\mathsf{cd}}
\newcommand{\ipd}{\mathsf{ipd}}

\renewcommand{\tilde}{\widetilde}

\begin{document}
\title{{On Closest Pair in Euclidean Metric:} \\ {Monochromatic is as Hard as Bichromatic}}
	\author{Karthik C.\ S.\thanks{Supported by Irit Dinur's ERC-CoG grant 772839 and BSF grant 2014371.}\vspace{0.03cm}\\\vspace{0.03cm} \textrm{{Weizmann Institute of Science}}\\ \texttt{karthik.srikanta@weizmann.ac.il} \and Pasin Manurangsi\thanks{Supported by NSF under Grants No. CCF 1655215 and CCF 1815434.}\vspace{0.03cm}\\\vspace{0.03cm} \textrm{University of California, Berkeley}\\ \texttt{pasin@berkeley.edu}}

	\date{}

	\maketitle
	\vspace{-0.5cm}
\begin{abstract}
Given a set of $n$ points in $\mathbb R^d$, the (monochromatic) \emph{Closest Pair} problem asks to find a pair of distinct points in the set that are closest in the $\ell_p$-metric. Closest Pair is a fundamental problem in Computational Geometry and understanding its fine-grained complexity in the Euclidean metric when $d=\omega(\log n)$ was raised as an open question in recent works  (Abboud-Rubinstein-Williams [FOCS'17], Williams~[SODA'18], David-Karthik-Laekhanukit [SoCG'18]).\vspace{0.1cm}

In this paper, we show that for every $p\in\mathbb R_{\ge 1}\cup\{0\}$, under the Strong Exponential Time Hypothesis (\SETH), for every  $\eps>0$, the following holds:
\begin{itemize}
\item No algorithm running in time $O(n^{2-\eps})$ can solve the Closest Pair problem in $d=(\log n)^{\Omega_{\varepsilon}(1)}$ dimensions in the $\ell_p$-metric.
\item There exists $\delta = \delta(\varepsilon)>0$ and $c = c(\varepsilon)\ge 1$ such that no algorithm running in time $O(n^{1.5-\eps})$ can approximate Closest Pair problem to a factor of $(1+\delta)$ in $d\ge c\log n$ dimensions in the $\ell_p$-metric.
\end{itemize}

In particular, our first result is shown by establishing the computational equivalence of the \emph{bichromatic} Closest Pair problem and the (monochromatic) Closest Pair problem (up to $n^{\varepsilon}$ factor in the running time) for $d=(\log n)^{\Omega_\varepsilon(1)}$ dimensions. \vspace{0.1cm}

Additionally, under \SETH, we rule out nearly-polynomial  factor approximation algorithms  running in subquadratic time for the (monochromatic) \emph{Maximum Inner Product} problem where we are given a set of $n$ points in $n^{o(1)}$-dimensional Euclidean space and are required to find a pair of distinct points in the set that maximize the inner product. \vspace{0.1cm}

At the heart of all our proofs is the construction of a dense bipartite graph with low \emph{contact dimension}, i.e., we construct a balanced bipartite graph on $n$ vertices with $n^{2-\varepsilon}$ edges whose vertices can be realized as points in a $(\log n)^{\Omega_\varepsilon(1)}$-dimensional Euclidean space  such that every pair of vertices which have an edge in the graph are at distance exactly 1 and every other pair of vertices are at distance greater than 1. This graph construction is inspired by the construction of locally dense codes introduced by Dumer-Miccancio-Sudan [IEEE Trans.\ Inf.\ Theory'03].
\end{abstract}
\clearpage

\tableofcontents
\clearpage

\section{Introduction}\label{sec:intro}

The Closest Pair of Points problem or {\em Closest Pair} problem (\CP) is a fundamental problem in computational geometry: given $n$ points in a $d$-dimensional metric space, find a pair of distinct points with the smallest distance between them. The Closest Pair problem for points in the Euclidean plane \cite{SH75,BS76} 
stands at the origins of the systematic study of the computational complexity of geometric problems
 \cite{PS85,Man89,KT05,CLRS09}. Since then, this problem has found abundant applications in geographic information systems \cite{H06}, clustering \cite{Zahn71,Alpaydin10}, and numerous matching problems (such as stable marriage \cite{WTFX07}).

The trivial algorithm for \CP examines every pair of points in the point-set and runs in time $O(n^2d)$. Over the decades, there have been a series of developments on \CP in low dimensional space for the Euclidean metric
\cite{Ben80,HNS88,KM95,SH75,BS76},
leading to a deterministic $O(2^{O(d)}n\log n)$-time algorithm
\cite{BS76} and a randomized $O(2^{O(d)}n)$-time algorithm
\cite{Rabin76,KM95}. For low (i.e., constant) dimensions, these algorithms are tight as
a matching lower bound of $\Omega(n\log n)$ was shown by Ben-Or
\cite{Ben83} and Yao \cite{Yao91} in the
{\em algebraic decision tree} model,
thus settling the complexity of \CP in low dimensions. On other hand, for very high dimensions (i.e., $d=\Omega(n)$)  there are
subcubic algorithms \cite{GS16,ILLP04} in the $\ell_1,\ell_2,$ and $\ell_\infty$-metrics using fast matrix multiplication algorithms \cite{L14}.
However, \CP in {medium dimensions}, i.e.,
$d=\text{polylog}(n)$, and in various $\ell_p$-metrics, 
have been a focus of study in machine learning and analysis of Big Data \cite{Kleinberg97}, and  it is surprising that, even with the tools and techniques that have
been developed over many decades, when $d=\omega(\log n)$,
there is no known subquadratic-time
(i.e., $O(2^{o(d)}n^{2-\varepsilon})$-time) algorithm,
for \CP in any standard distance measure
\cite{Indyk00,AC09,ILLP04}
. The absence of such algorithms was explicitly observed as early as the late nineties by Cohen and Lewis \cite{CL99} but there was not any explanation until recently. 

David, Karthik, and Laekhanukit \cite{DKL18} showed that for all $p>2$, assuming the \emph{Strong  Exponential Time Hypothesis} (\SETH), for every $\varepsilon>0$, no algorithm running in $n^{2-\varepsilon}$  time can solve \CP in the $\ell_p$-metric, even when $d=\omega(\log n)$.
Their conditional lower bound was based on the conditional lower bound (again assuming \SETH) of Alman and Williams \cite{AW15} for the \emph{Bichromatic Closest Pair} problem\footnote{We remark  that \BCP is of independent interest as it's equivalent to finding the \emph{Minimum Spanning Tree} in $\ell_p$-metric \cite{AESW91,KLN99}. Moreover, understanding the fine-grained complexity of \BCP has lead to better understanding of the query time needed for \emph{Approximate Nearest Neighbor} search problem (see Razenshteyn's thesis \cite{Raz17} for a survey about the problem) with polynomial preprocessing time \cite{R18}.} (\BCP) where we are given two sets of $n$ points in a $d$-dimensional metric space, and the goal is to find a pair of points, one from each set, with the smallest distance between them. Alman and Williams showed that for all $p\in\mathbb{R}_{\ge 1}\cup\{0\}$, assuming \SETH, for every $\varepsilon>0$,  no algorithm running in $n^{2-\varepsilon}$  time can solve \BCP in the $\omega(\log n)$-dimensional $\ell_p$-metric space. Given that $\cite{AW15}$ show their lower bound on \BCP for all $\ell_p$-metrics, the lower bound on \CP of \cite{DKL18} feels unsatisfactory, since the $\ell_2$-metric is arguably the most interesting metric to study \CP on.  On the other hand, the answer to the complexity of \CP in the Euclidean metric might be on the positive side, i.e., there might exist an
algorithm that performs well in the $\ell_2$-metric because there are more tools available, e.g., Johnson-Lindenstrauss’ dimension reduction \cite{JL84}. Thus we have the following question:

\begin{open}[Abboud-Rubinstein-Williams\footnote{Please see the erratum in \cite{ARW17a}.} \cite{ARW17}, Williams \cite{W18}, David
-Karthik-Laekhanukit \cite{DKL18}]\label{open:Q1}
Is there an algorithm running in time $n^{2-\varepsilon}$ for some $\varepsilon>0$ which can solve \CP in the Euclidean metric when the points are in $\omega(\log n)$ dimensions?
\end{open}

Even if the answer to the above question is negative, this does not rule out strong approximation algorithms for \CP in the Euclidean metric, which might suffice for all applications. Indeed, we do know of subquadratic approximation algorithms for \CP. For example, LSH based techniques can
solve $(1+\delta)$-\CP (i.e., $(1+\delta)$ factor approximate \CP) in $n^{2-\Theta\left(\delta\right)}$
time \cite{IM98}, but cannot do much better \cite{MNP07,OWZ14}.
In a recent breakthrough, Valiant \cite{V15} obtained an approximation algorithm for $(1+\delta)$-\CP with runtime of $n^{2-\Theta\left(\sqrt{\delta}\right)}$. The state of the art is an $n^{2-\tilde{\Theta}\left(\delta^{1/3}\right)}$-time
algorithm by Alman, Chan, and Williams \cite{ACW16}.
Can the dependence on $\delta$ be improved indefinitely? For the case of $(1+\delta)$-\BCP, assuming \SETH, Rubinstein \cite{R18} answered the question in the negative. Does $(1+\delta)$-\CP also admit the same negative answer?

\begin{open}\label{open:Q2}
Is there an algorithm running in time $n^{2-\varepsilon}$ for some $\varepsilon>0$ which can solve $(1+\delta)$-\CP  in the Euclidean metric when the points are in $\omega(\log n)$ dimensions for every $\delta>0$?
\end{open}

Another important geometric problem  is the \emph{Maximum Inner Product} problem (\MIP): given $n$ points in the $d$-dimensional Euclidean space, find a  pair of distinct points with the largest inner product. 
This problem along with its bichromatic variant (\emph{Bichromatic Maximum Inner Product} problem, denoted \BMIP) is extensively studied in literature (see \cite{ARW17} and references therein). Abboud, Rubinstein, and Williams \cite{ARW17} showed that assuming \SETH, for every $\varepsilon>0$, no $2^{(\log n)^{1-o(1)}}$-approximation algorithm running in $n^{2-\varepsilon}$  time can solve \BMIP when $d=n^{o(1)}$. It is a natural question to ask if their inapproximability result can be extended to \MIP:

\begin{open}\label{open:Q3}
Is there an algorithm running in time $n^{2-\varepsilon}$ for some $\varepsilon>0$ which can solve $\gamma$-\MIP in $n^{o(1)}$ dimensions for even $\gamma=2^{(\log n)^{1-o(1)}}$?
\end{open}

\subsection{Our Results}\label{sec:result}

In this paper we address all three previously mentioned open questions. First, we almost completely resolve Open~Question~\ref{open:Q1}. In particular, we show the following.

\begin{theorem}[Subquadratic Hardness of \CP; Informal, See Theorem~\ref{thm:CP-01}]\label{thm:CP}
Let $p\in\mathbb{R}_{\ge 1}\cup\{0\}$. Assuming \SETH, for every $\varepsilon>0$, no algorithm running in $n^{2-\varepsilon}$  time can solve \CP in the $\ell_p$-metric, even when $d=\left(\log n\right)^{\Omega_{\varepsilon}(1)}$.
\end{theorem}

In particular we would like to emphasize that the dimension for which we show the lower bound on \CP depends on $\varepsilon$. We would also like to remark that our lower bound holds even when the input point-set of \CP is a subset of $\{0,1\}^d$. Finally, we note that the centerpiece of the proof of the above theorem (and also the proofs of the other results that will be subsequently mentioned) is the  construction of a dense bipartite graph with low \emph{contact dimension}, i.e., we construct a balanced bipartite graph on $n$ vertices with $n^{2-\varepsilon}$ edges whose vertices can be realized as points in a $(\log n)^{\Omega_\varepsilon(1)}$-dimensional $\ell_p$-metric space  such that every pair of vertices which have an edge in the graph are at distance exactly 1 and every other pair of vertices are at distance greater than 1. This graph construction is inspired by the construction of locally dense codes introduced by Dumer, Miccancio, and Sudan \cite{DMS03} and uses special density properties of Reed Solomon codes. A detailed proof overview is given in Section~\ref{sec:overviewRS}.

Next, we improve our result in Theorem~\ref{thm:CP} in some aspects by showing $1+o(1)$ factor inapproximability of \CP even in \emph{$O_\varepsilon(\log n)$} dimensions, but can only rule out algorithms running in $n^{1.5-\varepsilon}$ time (as opposed to Theorem~\ref{thm:CP} which rules out exact algorithms for \CP running in $n^{2-\varepsilon}$ time). More precisely, we show the following.

\begin{theorem}[Subquadratic Hardness of gap-\CP]\label{thm:gapCP}
Let $p\in\mathbb{R}_{\ge 1}\cup\{0\}$. Assuming \SETH, for every $\varepsilon>0$, there exists $\delta(\varepsilon)>0$ and $c(\varepsilon)>1$ such that  no algorithm running in $n^{1.5-\varepsilon}$  time that can solve $(1+\delta)$-\CP in the $\ell_p$-metric,  even when $d=c\log n$.
\end{theorem}

We remark that the  $n^{1.5-\varepsilon}$ lower bound on approximate \CP is an artifact of our proof strategy and that a different approach or an improvement in the state-of-the-art bound on the number of minimum weight codewords in algebraic geometric codes (which are used in our proof), will lead to the complete resolution of  Open~Question~\ref{open:Q2}.

It should also be noted that the approximate version of \CP and the dimension are closely related. Namely, using standard dimensionality reduction techniques~\cite{JL84}\footnote{In fact, since our results applies to $\{0, 1\}$-vectors, simply subsampling coordinates would also work.} for $(1 + \delta)$-\CP, one can always assume that $d = O_\delta(\log n)$. In other words, hardness of $(1 + \delta)$-\CP immediately yields logarithmic dimensionality bound as a byproduct.

Finally, we completely answer Open~Question~\ref{open:Q3} by showing the following inapproximability result for \MIP, matching the hardness for \BMIP from~\cite{ARW17}.

\begin{theorem}[Subquadratic Hardness of gap-\MIP]\label{thm:MIP}
Assuming \SETH, for every $\varepsilon>0$,  no algorithm running in $n^{2-\varepsilon}$  time  can solve $\gamma$-\MIP for any $\gamma\le 2^{(\log n)^{1-o(1)}}$, even when $d=n^{o(1)}$.
\end{theorem}

Recently, there have been a lot of results connecting \BCP or $(1+o(1))$-\BCP to other problems (see \cite{R18,C18,C18a,CW19}). Now such connections can be extended to \CP as well. For example, the following conditional lower bound follows from \cite{R18} for gap-\CP in the edit distance metric and for completeness a proof is given in Appendix~\ref{sec:edit}.

\begin{theorem}[Subquadratic Hardness of gap-\CP in edit distance metric]\label{thm:CP-Edit}
Assuming \SETH, for every $\varepsilon>0$, there exists $\delta(\varepsilon)>0$ and $c(\varepsilon)>1$ such that  no algorithm running in $n^{1.5-\varepsilon}$  time   can solve $(1+\delta)$-\CP in the edit distance metric,  even when $d=c\log n\log\log n$.
\end{theorem}

\section{Proof Overview}\label{sec:overview}

In this section, we provide an overview of our proofs. For ease of presentation, we will sometimes be informal here; all notions and proofs are formalized in subsequent sections. Our overview is organized as follows. First, in Subsection~\ref{sec:overviewRS}, we outline our proof of running time lower bounds for exact \CP (Theorem~\ref{thm:CP}). Then, in Subsection~\ref{subsec:abstract}, we abstract part of our reduction using error-correcting codes, and relate them back to the works on locally dense codes~\cite{DMS03,ChengW12,M14} that inspire our constructions. Finally, in Subsection~\ref{sec:overviewAG}, we briefly discuss how to modify the base construction (i.e. code properties) to give conditional lower bounds for approximate \CP and \MIP (Theorems~\ref{thm:gapCP} and~\ref{thm:MIP}).

\subsection{Conditional Lower Bound on Exact Closest Pair}\label{sec:overviewRS}

In this subsection, we provide a proof overview of a slightly weaker version of Theorem~\ref{thm:CP},  i.e., we show that assuming \SETH, for every $p\in\mathbb{R}_{\ge 1}\cup\{0\}$, no subquadratic time algorithm can solve \CP in the $\ell_p$-metric when $d=(\log n)^{\omega(1)}$. We prove such a result by reducing \BCP in dimension $d$ to \CP in dimension $d+(\log n)^{\omega(1)}$, and the subquadratic hardness for \CP follows from the subquadratic hardness of \BCP established by \cite{AW15}. 
Note that the results in this paper remain interesting even if \SETH is false, as our reduction shows that \BCP and \CP  are computationally equivalent\footnote{We can reduce an instance of \CP to an instance of \BCP by randomly partitioning the input set of \CP instance into two, and the optimal closest pair of points will be in different sets with probability $\nicefrac{1}{2}$ (and this reduction can be made deterministic).} (up to $n^{o(1)}$ factor in the running time) when $d=(\log n)^{\omega(1)}$. The conditional lower bound on \CP is merely a consequence of this computational equivalence. Finally, we note that a similar equivalence also holds between \MIP and \BMIP. 

\paragraph{Understanding an obstacle of \cite{DKL18}.} Our proof builds on the ideas of \cite{DKL18} who showed that assuming \SETH, for every $p>2$, no subquadratic time algorithm can solve \CP in the $\ell_p$-metric when $d=\omega(\log n)$. They did so by connecting the complexity of \CP and \BCP via the \emph{contact dimension} of the balanced complete bipartite graph (biclique), denoted by $K_{n,n}$. We elaborate on this below.

To motivate the idea behind~\cite{DKL18}, let us first consider the trivial reduction from \BCP to \CP: given an instance $A, B$ of \BCP, we simply output $A \cup B$ as an instance of \CP. This reduction fails because there is no guarantee on the distances of a pair of points both in $A$ (or both in $B$). That is, there could be two points $\ba, \ba' \in A$ such that $\|\ba - \ba'\|_p$ is much smaller than the optimum of \BCP on $A, B$. If we simply solve \CP on $A \cup B$, we might find such $\ba, \ba'$ as the optimal pair but this does not give the answer to the original \BCP problem. In order to circumvent this issue, one needs a gadget that ``stretch'' pairs of points both in $A$ or both in $B$ further apart while keeping the pairs of points across $A$ and $B$ close (and preserving the optimum of \BCP on $A,B$). It turns out that this notion corresponds exactly to the contact dimension of the biclique, which we define below.

\begin{definition}[Contact Dimension~\cite{P80}] \label{def:cd-informal}
For any graph $G = (V, E)$, a mapping $\tau: V \to \mathbb{R}^d$ is said to \emph{realize $G$ (in the $\ell_p$-metric)} if for some $\beta > 0$, the following holds for every distinct vertices $u, v$:
\begin{align}
\|\tau(u) - \tau(v)\|_p &= \beta \text{ if } \{u, v\} \in E, \text{ and,} \label{eq:close-points} \\
\|\tau(u) - \tau(v)\|_p &> \beta \text{ otherwise.} \label{eq:stretch}
\end{align}
The \emph{contact dimension} (in the $\ell_p$-metric) of $G$, denoted by $\cd_p(G)$, is the minimum $d \in \mathbb{N}$ such that there exists $\tau: V \to \mathbb{R}^d$ realizing $G$ in the $\ell_p$-metric.
\end{definition}

In this paper, we will be mainly interested in the contact dimension of bipartite graphs. Specifically,~\cite{DKL18} only consider the contact dimension of the biclique $K_{n, n}$. Notice that a realization of biclique ensures that vertices on the same side are far from each other while vertices on different sides are close to each other preserving the optimum of \BCP; these are exactly the desired properties of a gadget outlined above. Using this, \cite{DKL18} give a reduction from \BCP to \CP which shows that the two are computationally equivalent whenever $d=\Omega(\cd_p(K_{n,n}))$, as follows.

Let $A,B\subseteq \mathbb{R}^d$ each of cardinality $n$ be an instance of \BCP and let $\tau: A\dot\cup B \to \mathbb{R}^{\cd_p(K_{n,n})}$ be a map realizing the biclique $(A\dot\cup B, A \times B)$ in the $\ell_p$-metric; we may assume w.l.o.g. that $\beta = 1$. Let $\delta$ be the distance between any point in $A$ and any point in $B$ (i.e., $\delta$ is an upper bound on the optimum of \BCP). Let $\rho>0$ be such that $\|\tau(\ba) - \tau(\bb)\|_p > 1+\rho$ for all $\ba \in A, \bb \in B$ (and this is guaranteed to exist by~\eqref{eq:stretch}). Moreover, let $k > \delta/\rho$ be any sufficiently large number. Consider the point-sets $\tilde A,\tilde B\subseteq \mathbb{R}^{d+\cd_p(K_{n,n})}$ of cardinality $n$ each defined as
$$
\tilde A=\{\ba\circ (k \cdot \tau(\ba))\mid \ba \in A\},\ \tilde B=\{\bb \circ (k \cdot \tau(\bb)) \mid \bb \in B\},
$$
where $\circ$ denotes the concatenation between two vectors and $k \cdot \bx$ denotes the usual scalar-vector multiplication (i.e. scaling $\bx$ up by a factor of $k$). For brevity, we write $\tilde{\ba}$ and $\tilde{\bb}$ to denote $\ba\circ (k \cdot \tau(\ba))$ and $\bb \circ (k \cdot \tau(\bb))$ respectively.

We now argue that, if we can find the closest pair of points in $\tA \cup \tB$, then we also immediately solve \BCP for $(A, B)$. More precisely, we claim that $(\ba^*, \bb^*) \in A \times B$ is a bichromatic closest pair of $(A, B)$ if and only if $(\tilde{\ba^*}, \tilde{\bb^*})$ is a closest pair of $\tA \cup \tB$.

To see that this is the case, observe that, for cross pairs $(\tilde{\ba}, \tilde{\bb}) \in \tA \times \tB$, ~\eqref{eq:close-points} implies that the distance $\|\tilde{\ba} - \tilde{\bb}\|_p$ is exactly $(k^p + \|\ba - \bb\|^p_p)^{1/p}$; hence, among these pairs, $(\tilde{\ba^*}, \tilde{\bb^*})$ is a closest pair iff $(\ba^*, \bb^*)$ is a bichromatic closest pair in $A, B$. Notice also that, since the bichromatic closest pair in $A, B$ is of distance at most $\delta$, the closest pair in $\tA \cup \tB$ is of distance at most $(k^p + \delta^p)^{1/p} \leq k + \delta$.

On the other hand, for pairs both from $\tA$ or both from $\tB$, the distance must be at least $k(1 + \rho)$, which is more than $k + \delta$ from our choice of $k$. As a result, these pairs cannot be a closest pair in $\tA \cup \tB$, and this concludes the sketch of the proof.

There are a couple of details that we have glossed over here: one is that the gap $\rho$ cannot be too small (e.g., $\rho$ cannot be as small as $\nicefrac{1}{2^n}$) and the other is that we should be able to construct $\tau$ efficiently. Nevertheless, these are typically not an issue.

\cite{DKL18} show that $\cd_p(K_{n,n})=\Theta(\log n)$ when $p>2$ and that the realization can be constructed efficiently and with sufficiently large $\rho$. This implies the subquadratic hardness of \CP (by reduction from \BCP) in the $\ell_p$-metric for all $p>2$ and $d=\omega(\log n)$. However, it was known that $\cd_2(K_{n,n})=\Theta(n)$ \cite{FM88}. Thus, they could \emph{not} extend their conditional lower bound to \CP in the Euclidean metric\footnote{Note that plugging in the bound on $\cd_2(K_{n,n})$ in the result of \cite{DKL18} yields that assuming \SETH, no subquadratic in $n$ running time algorithm can solve \CP when $d=\Omega(n)$. This is not a meaningful lower bound as just the input size of \CP when $d=\Omega(n)$ is $\Omega(n^2)$.} even when $d=o(n)$. In fact, this is a serious obstacle as it rules out many natural approaches to reduce \BCP to \CP in a black-box manner. Elaborating, the lower bound on $\cd_2(K_{n,n})$ rules out local gadget reductions which would replace each point with a composition of that point and a gadget with a small increase in the number of dimensions, as such gadgets can be used to construct a realization of $K_{n,n}$ in the Euclidean metric in a low dimensional space, contradicting the lower bound on $\cd_2(K_{n,n})$.

\paragraph{Overcoming the Obstacle: Beyond Biclique.} \begin{sloppypar}We overcome the above obstacle by considering dense bipartite graphs, instead of the biclique. More precisely, we  show that there exists a balanced bipartite graph $G^* = (A^* \dot\cup B^*, E^*)$ on $2n$ vertices such that $|E^*|~\geq~n^{2 - o(1)}$ and $\cd_p(G^*)$ is small (i.e. $\cd_p(G^*) \leq (\log n)^{\omega(1)}$). We  give a construction of such a graph below but before we do so, let us briefly argue why this suffices to show that \BCP and \CP are computationally equivalent (up to $n^{o(1)}$ multiplicative overhead in the running time) for dimension $d = \Omega(\cd_p(G^*))$.\end{sloppypar}

Let us consider the same reduction which produces $\tA, \tB$ as before, but instead of using a realization of the biclique, we use a realization $\tau$ of $G^*$. This reduction is of course incorrect: if $(\ba^*, \bb^*)$ is not an edge in $G^*$, then $\|\tau(\ba^*) - \tau(\bb^*)\|_p$ could be large and, thus the corresponding pair of points $(\tilde{\ba^*}, \tilde{\bb^*})\in \tA\times \tB$, may not be the closest pair. Nevertheless, we are not totally hopeless: if $(\ba^*, \bb^*)$ is an edge, then we are in good shape and the reduction is correct.

With the above observation in mind, consider picking a random permutation $\pi$ of $A \cup B$ such that $\pi(A) = A$ and $\pi(B) = B$ and then initiate the above reduction with the map $(\tau \circ \pi)$ instead of $\tau$. Note that $\tau \circ \pi$ is simply a realization of an appropriate permutation $G'$ of $G^*$ (i.e., $G'$ is isomorphic to $G^*$). Due to this, the probability that we are ``lucky'' and  $(\ba^*, \bb^*)$ is an edge in $G'$ is $p := |E|/n^2$; when this is the case, solving \CP on the resulting instance would give the correct answer for the original \BCP instance. If we repeat this $\log n / p = n^{o(1)}$ times, we would find the optimum of the original \BCP instance with high probability.

To recap, even when $G^*$ is not a biclique, we can still use it to give a reduction from \BCP to \CP, except that the reduction  produces multiple (i.e. $\tilde{O}(n^2/|E^*|)$) instances of \CP. We remark here that the reduction can be derandomized: we can deterministically (and efficiently) pick the permutations so that the permuted graphs covers $K_{n, n}$ (see Lemma~\ref{lem:cover}). As a minor digression, we would like to draw a parallel here with a recent work of Abboud, Rubinstein, and Williams \cite{ARW17}. The obstacle raised in \cite{DKL18} is about the impossibility of certain kinds of many-one gadget reductions. We overcame it by designing a reduction from \BCP to \CP which not only increased the number of dimensions but also the number of points (by creating multiple instances of \CP). 
This technique is also utilized in \cite{ARW17} where they showed the impossibility of Deterministic Distributed PCPs (Theorem~I.2 in \cite{ARW17}) but then overcame that obstacle by using an advice (which is then enumerated over resulting in multiple instances) to build Non-deterministic Distributed PCPs.

\paragraph{Constructing a dense bipartite graph with low contact dimension.}
We now proceed to  construct the desired graph $G^* = (A^* \cup B^*, E^*)$.  Note that any construction of a dense bipartite graph with contact dimension $n^{o(1)}$ is non-trivial. This is because it is known that a random graph has contact dimension $\Omega(n)$ in the Euclidean metric with high probability \cite{RRS89,BL05}, and therefore our graph construction must be significantly better than a random graph. 

Our realization $\tau^*$ of $G^*$ will map into a subset of $\{0,1\}^{(\log n)^{\omega(1)}}$. As a result, we can fix $p=0$, since a realization of a graph with entries in $\{0,1\}$ in the Hamming-metric also realizes the same graph in every $\ell_p$-metric for any $p\neq\infty$. 

\begin{sloppypar}Fix $g=\omega(1)$. We associate $[n]$ with $\mathbb{F}_q^h$ where $q=\Theta\left((\log n)^g\right)$ is a prime and $h~=~\Theta\left(\frac{\log n}{g\cdot \log\log n}\right)$. Let $\mathcal{P}$ be the set of all univariate polynomials (in $x$) over $\mathbb{F}_q$ of degree at most $h-1$. We have that $|\mathcal{P}|=q^h=n$ and associate $\mathcal{P}$ with $A^*$. Let $\mathcal{Q}$ be the set of all univariate monic polynomials (in $x$) over $\mathbb{F}_q$ of degree $h$, i.e.,\end{sloppypar}
$$
\mathcal{Q}=\{x^h+p(x)\mid p(x)\in\mathcal{P}\}.
$$
We associate the polynomials in $\mathcal{Q}$ with the vertices in $B^*$ (note that $|\mathcal{Q}|=n$). In fact, we view the vertices in $A^*$ and $B^*$ as being uniquely labeled by polynomials in $\mathcal{P}$ and $\mathcal{Q}$ respectively. For notational clarity, we write $p_a$ (resp. $p_b$) to denote the polynomial in $\cP$ (resp. $\cQ$) that is associated to $a \in A^*$ (resp. $b \in B^*$).

For every $a\in A^*$ and $b\in B^*$, we include $(a,b)$ as an edge in $E^*$ if and only if the polynomial $p_b-p_a$ (which is of degree $h$) has $h$ distinct roots. This completes the construction of $G^*$. We have to now show the following two claims about $G^*$: (i) $|E^*| =n^{2-O\left(\nicefrac{1}{g}\right)}=n^{2-o(1)}$ and (ii) there is $\tau: A^* \dot\cup B^* \to \{0,1\}^{(\log n)^{O(g)}}=\{0,1\}^{(\log n)^{\omega(1)}}$ that realizes $G^*$.

To show (i), let $\mathcal{R}$ be the set of all monic polynomials of degree $h$ with $h$ distinct roots. We have that $|\mathcal R|=\binom{q}{h}$. Fix a vertex $a\in A^*$. Its degree in $G^*$ is exactly $|\mathcal R|=\binom{q}{h}$. This is because, for every polynomial $r\in \mathcal{R}$, $r+a$ belongs to $\mathcal{Q}$, and therefore $(a,r+a)\in E^*$. This implies the following bound on $|E^*|$:
$$
|E^*| = q^h \cdot \binom{q}{h} \ge q^h\cdot \frac{q^h}{h^h} > \frac{n^2}{(\log n)^{\Theta((\log n)/(g\cdot \log\log n)})}=n^{2-O\left(\nicefrac{1}{g}\right)}.
$$

Next, to show (ii), we construct a realization $\tau^*: A^* \dot\cup B^* \to \mathbb{F}_q^q$ of $G^*$. We note that, it is simple to translate the entries to $\{0, 1\}$ instead of $\mathbb{F}_q$, by replacing $i \in \mathbb{F}_q$ with the $i$-th standard basis $\be_i \in \{0, 1\}^q$. This would result in a realization $\tau^*: A^* \dot\cup B^* \to \{0, 1\}^{q^2}$ of $G^*$; notice that the dimension of $\tau^*$ is $q^2 = \Theta((\log n)^{2g})$ as claimed.

We define $\tau^*$   as follows.
\begin{itemize}
\item For every $a \in A^*$, $\tau^*(a)$ is simply the vector of evaluation of $p_a$ on every element in $\mathbb{F}_q$. More precisely, for every $j \in [q]$, the $j$-th coordinate of $\tau^*(a)$ is $p_a(j - 1)$.
\item Similarly, for every $b \in B^*$ and $j \in [q]$,  the $j$-th coordinate of $\tau^*(b)$ is $p_b(j - 1)$.
\end{itemize}

We  now show that $\tau^*$ is indeed a realization of $G^*$; specifically, we show that $\tau^*$ satisfies~\eqref{eq:close-points} and~\eqref{eq:stretch} with $\beta = q - h$.

Consider any edge $(a, b) \in E^*$. Notice that $\|\tau^*(a)-\tau^*(b)\|_0$ is the number of $x\in\mathbb{F}_q$ such that $p_b(x) - p_a(x) \neq 0$. By definition of $E^*$, $p_b - p_a$ is a polynomial with $h$ distinct roots over $\mathbb{F}_q$. Thus, $\|\tau^*(a)-\tau^*(b)\|_0 = q - h = \beta$ as desired.

Next, consider a non-edge $(a, b) \in (A^*\times B^*) \setminus E^*$ . Then, we know that $p_b - p_a$ has at most $h - 1$ distinct roots over $\mathbb{F}_q$. Therefore, the polynomial $p_b - p_a$ is non-zero on at least $q-h+1$ coordinates. This implies that $\|\tau^*(a) - \tau^*(b)\|_0 \ge q-h+1 > \beta$.

Finally, for any distinct $a, a'\in A^*$, we have $\|\tau^*(a) - \tau^*(a')\|_0 \ge q-h+1$ because $p_a-p_{a'}$ is a non-zero polynomial of degree at most $h-1$ and thus can be zero over $\mathbb F_q$ in at most $h-1$ locations. Similarly, $\|\tau^*(b) - \tau^*(b')\|_0 \ge q-h+1$ for any distinct $b, b' \in B^*$.

This completes the proof sketch for both the claims about $G^*$ and yields Theorem~\ref{thm:CP} for $d=(\log n)^{\omega(1)}$. Finally we remark that in the actual proof of Theorem~\ref{thm:CP}, we will set the parameters in the above construction more carefully and achieve the bound $\cd_p(G^*) = (\log n)^{O_{\varepsilon}(1)}$.

\subsection{Abstracting the Construction via Error-Correcting Codes} \label{subsec:abstract}

Before we move on to discuss the proofs of Theorems~\ref{thm:MIP} and~\ref{thm:gapCP}, let us give an abstraction of the construction in the previous subsection. This will allow us to easily generalize the construction for the aforemention theorems, and also to explain where our motivation behind the construction comes from in the first place.

\paragraph{Dense Bipartite Graph with Low Contact Dimension from Codes.}
\begin{sloppypar}In order to construct a balanced bipartite graph $G^*$ on $2n$ vertices with $n^{2-o(1)}$ edges such that $\cd_p(G^*)\le~d^*$, it suffices to have a code $C^*$ with the following properties (for code-related definitions, see Section~\ref{sec:code-prelim}):\end{sloppypar}
\begin{itemize}
\item $C^* \subseteq \mathbb{F}_q^\ell$ of cardinality $n$ is a linear code with block length $\ell$ over alphabet $\mathbb{F}_q$, and minimum distance $\Delta$.
\item There exists a \emph{center} $s^*\in \mathbb{F}_q^\ell$ and $r^*<\Delta$ such that  $|C^*|^{1-o(1)}$ codewords are at Hamming distance exactly $r^*$ from $s^*$ and no codeword is at distance less than $r^*$ from $s^*$. 
\item $q\cdot \ell = d^*$.
\end{itemize} 
We also require that $C^*$ and $s^*$ can be constructed in $\poly(n)$ time but we shall ignore this requirement for the ease of exposition.

We describe  below how to construct $G^*$ from $C^*$, but first note that the construction of $G^*$ we saw in the previous subsubsection was just showing that Reed Solomon codes \cite{RS60} of block length $q=\Theta((\log n)^g)$ and message length $h=\Theta\left(\frac{\log n}{g\cdot \log\log n}\right)$ over alphabet $\mathbb F_q$ with minimum distance $q-h+1$ has the above properties. The center $s^*$ in that construction was the evaluation of the polynomial $x^h$ over $\mathbb{F}_q$, and $r^*$ was $q-h$.

In general, to construct $G^*$ from $C^*$, we first define a subset $S^*\subseteq\mathbb{F}_q^\ell$ of cardinality $n$ as follows:
$$
S^*=\{\bs^*+ \bc \mid \bc \in C^*\}.
$$
We associate the vertices in $A^*$ with the codewords of $C^*$ and vertices in $B^*$ with the strings in $S^*$. For any $(\ba, \bb)\in A^*\times B^*$, let $(\ba, \bb)\in E^*$ if and only if $\|\bb - \ba\|_0=r^*$. This completes the construction of $G^*$. We have to now show the following claims about $G^*$: (i) $|E^*| = n^{2-o(1)}$ and (ii) there is $\tau: A^*\dot\cup B^* \to \{0,1\}^{q\cdot \ell}$ that realizes $G^*$.

Item (i) follows rather easily from the properties of $C^*$ and $s^*$. Let $T^*$ be the subset of $C^*$  of all codewords which are at distance exactly equal to $r^*$ from $s^*$. From the definition of $s^*$, we have $|T^*|=|C^*|^{1-o(1)}$. Fix $\ba\in A^*$. Its degree in $G^*$ is $|T^*|=|C^*|^{1-o(1)}$. This is because for every codeword $\bt\in T^*$ we have that $\bt-\ba$ is a codeword in $C^*$ (from the linearity of $C^*$) and thus $\bs^*-\bt+\ba$ is in $S^*$, and therefore $(\ba,\bs^*-\bt+\ba)\in E^*$. 

For item (ii), consider the identity mapping $\tau^*: A^* \dot\cup B^* \to \mathbb{F}_q^\ell$ that maps each string to itself. It is simple to check that $\tau^*$ realizes $G^*$ in the Hamming metric (with $\beta = r^*$).

Recall from the previous subsection that given $\tau^*: A^* \dot\cup B^* \to \mathbb{F}_q^\ell$ that realizes $G^*$ in the Hamming metric, it is easy to construct $\tau: A^* \dot\cup B^* \to \{0, 1\}^{q \cdot \ell}$ that realizes $G^*$ in the Hamming metric with a $q$ multiplicative factor blow-up in the dimension. This completes the proof of both the claims about $G^*$ and gives a general way to prove Theorem~\ref{thm:CP} given the construction of $C^*$ and $s^*$. 

\paragraph{Finding Center from Another Code.} One thing that might not be clear so far is: where does the center $s^*$ come from? Here we provide a systematic way to produce such an $s^*$, by looking at another code that contains $C^*$. More precisely, let $C^* \subseteq \tilde{C}^* \subseteq \mathbb{F}_q^{\ell}$ be two linear codes with the same block length and alphabet. Suppose that the distance of $C^*$ is $\Delta$, the distance of $\tilde{C}^*$ is $r^*$ and that $r^* < \Delta$. It is easy to see that, by taking $s^*$ to be any element of $\tilde{C}^* \setminus C^*$, it holds that every codeword in $C^*$ is at distance at least $r^*$ from $s^*$, simply because both $s^*$ and the codewords of $C^*$ are codewords of $\tilde{C}^*$. 

Hence, we are only left to argue that there are many codewords of $C^*$ that is of distance exactly $r^*$ from $s^*$. While this is not true in general, we can show by an averaging argument that this is true (for some $s^* \in \tilde{C}^*$) if a large fraction (e.g. $|C^*|^{-o(1)}$ fraction) of codewords of $\tilde{C}^*$ has Hamming weight exactly $r^*$  (see Lemma~\ref{lem:finding-center}).

Indeed, viewing in this light, our previous choice of center for Reed-Solomon code (i.e. evaluation of $x^h$) is not coincidental: we simply take $\tilde{C}^*$ to be another Reed-Solomon code with message length $h + 1$ (whereas the base code $C^*$ is of message length $h$). 

\paragraph{Comparison to Locally Dense Codes.}
We end this subsection by remarking that the codes that we seek are very similar to locally dense codes \cite{DMS03,ChengW12,M14}, which is indeed our inspiration. A \emph{locally dense code} is a linear code of block length $\ell$ and large minimum distance $\Delta$, admitting a ball  centered at $s$ of radius\footnote{
Clearly, for the ball to contain more than a single codeword, it must be  $r\ge \Delta/2$. Here we are interested in balls with radius not much bigger than that, say  $r<\gamma\cdot \Delta$  for some constant  $1/2<\gamma<1$.}  $r < \Delta$ and containing a large (i.e. $\exp(\poly(\ell))$) number of codewords\footnote{Strictly speaking, a locally dense code also requires an auxiliary matrix $T$  used to index these codewords. However, in previous works, finding $T$ is typically not hard given the center $s$. Hence, we ignore $T$ in our discussion here for the ease of exposition.}. Such codes are non-trivial to construct and in particular all known constructions of locally dense codes are using codes that beat the 
Gilbert-Varshamov (GV) bound \cite{G52,V57}; in other words we need to do better than random codes to construct them. 
This is because (as noted in~\cite{DMS03}), for a random code $C\subseteq \mathbb F_q^\ell$ (or any code that does not beat the GV bound), a random point in $\mathbb{F}_q^\ell$ acting as the center contains in expectation less than one codeword in a ball of radius $\Delta$. Of course, this is simply an intuition and not a formal proof that a locally dense code needs to beat the GV bound, since there may be more sophisticated ways to pick a center.

Although the codes we require are similar to locally dense codes, there are differences between the two. Below we list four such differences: the first two makes it \emph{harder} for us to construct our codes whereas the latter two makes it \emph{easier} for us.
\begin{itemize}
\item We seek a center $s^*$ so that no codewords in $C^*$ lies at distance less than $r^*$, as opposed to locally dense codes which allows codewords to be close to $s^*$. This is indeed where our idea of using another code $\tilde{C}^* \supseteq C^*$ comes in, as picking $s^*$ from $\tilde{C}^* \setminus C^*$ ensures us that no codeword of $C^*$ is too close to $s^*$.
\item Another difference is that we need the number of codewords at distance $r^*$ from $s^*$ to be very large, i.e., $|C^*|^{1 - o(1)}$, whereas locally dense codes allow for much smaller number of codewords. Indeed, the deterministic constructions from~\cite{ChengW12,M14} only yield the bound of $2^{O(\sqrt{\log |C^*|})}$. Hence, these do not directly work for us.
\item Locally dense codes requires $r$ to be at most $(1 - \varepsilon) \Delta$ for some constant $\varepsilon > 0$, whereas we are fine with any $r^* < \Delta$. In fact, our Reed-Solomon code based construction above only yields $r^* = \Delta - 1$ which would not suffice for locally dense codes. Nevertheless, as we will see later for inapproximability of \CP, we will also need the ratio $r^*/\Delta$ to be a constant bounded away from 1 as well and, since we need a code with these extraordinary properties, they are very hard to find. Indeed, in this case we only manage to prove a weaker lower bound on gap-\CP.
\item \begin{sloppypar}Finally, we remark that locally dense codes are required to be efficiently constructed in $\poly(\log |C^*|)$ time, which is part of why it is hard to find. Specifically, while~\cite{DMS03} shows that an averaging argument works for a random center, derandomizing this is a big issue and a few subsequent works are dedicated solely to this issue~\cite{ChengW12,M14}. (We also note that it remains open whether a center can be deterministically found for a variant of locally dense codes used in hardness of parameterized version of the minimum distance problem. See~\cite{BGKM18} for more details.)
On the other hand, brute force search (over all codewords in $\tilde{C}^*$) suffices to find a center for us, as we are allowed construction time of $\poly(|C^*|)$.\end{sloppypar}
\end{itemize}

\subsection{Inapproximability of Closest Pair and Maximum Inner Product}\label{sec:overviewAG}

In this subsection, we sketch our inapproximability results for \MIP and \CP. Both  these results use the same reduction that we had from \BCP to \CP, except that we now need stronger properties from the gadget, i.e., the previously used notions of contact dimension does not suffice anymore. Below we sketch the required strengthening of the  gadget properties   and explain how to achieve them.

\subsubsection{Approximate Maximum Inner Product} \label{sec:mip-overview}

Observe that the gadget we construct for \CP in Subsection~\ref{subsec:abstract} can also be written in terms of inner product as follows: there exists a dense balanced bipartite graph $G^* = (A^* \dot\cup B^*, E^*)$, a mapping $\tau: A^* \dot\cup B^* \to \{0, 1\}^{q\cdot \ell}$ such that the following holds. 
\begin{enumerate}[(i)]
\item For all edges $(a, b) \in E^*$, $\left<\tau(a), \tau(b)\right> = \ell - r^*$. \label{list:edge}
\item For all edges $(a, b) \in (A^* \times B^*) \setminus E^*$, $\left<\tau(a), \tau(b)\right> < \ell - r^*$. \label{list:same}
\item For all distinct $a,b$ both from $A^*$ or both from $B^*$, $\left<\tau(a), \tau(b)\right> \leq \ell - \Delta$. \label{list:cross}
\end{enumerate}
Notice that we wrote the conditions above in a slightly different way than in previous subsections; previously in the contact dimension notation,~(\ref{list:same}) and~(\ref{list:cross}) would be simply written together as: for all non-edge $(a, b)$, $\left<\tau(a), \tau(b)\right> < \ell - r^*$. This change is intentional, since, to get gap in our reductions, we only need a gap between the bounds in~(\ref{list:edge}) and~(\ref{list:cross}) (but not in (\ref{list:same})). In particular, to get hardness of approximating \MIP, we require $\frac{\ell-r^*}{\ell-\Delta}$ to be at least $(1 + \varepsilon)$ for some $\varepsilon > 0$.

From our Reed-Solomon construction above, $\ell - \Delta$ and $\ell - r^*$ are exactly the message length of $C^*$ minus one and the message length of $\tilde{C}^*$ minus one respectively. Previously, we selected these two to be $h$ and $h + 1$. Now to obtain the desired gap, we simply take the larger code $\tilde{C}^*$ to be a Reed-Solomon code with larger (i.e. $(1 + \varepsilon)h$) message length\footnote{This approach can in fact give not just $(1 + \varepsilon)$ but arbitrarily large constant gap between the two cases. In the actual reduction, we take this gap to be 3 (Theorem~\ref{thm:MIP-pre-tensor}), which makes some computations simpler.}.

Finally, we note that even with the above gadget, the reduction only gives a small (i.e. $1 + o(1)$) factor hardness of approximating \MIP (Theorem~\ref{thm:MIP-pre-tensor}). To boost the gap to near polynomial, we simply tensor the vectors with themselves (see Section~\ref{sec:MIP}).

\subsubsection{Approximate Closest Pair} \label{sec:apx-cp-overview}

Once again, recall that we have the following gadget from Subsection~\ref{subsec:abstract}: there exists a dense balanced bipartite graph $G^* = (A^* \dot\cup B^*, E^*)$, a mapping $\tau: A^* \dot\cup B^* \to \{0, 1\}^{q\cdot \ell}$ such that the following holds. 
\begin{enumerate}[(i)]
\item For all edges $(a, b) \in E^*$, $\|\tau(a)- \tau(b)\|_0 =r^*$. \label{list:edge-cp}
\item For all edges $(a, b) \in (A^* \times B^*) \setminus E^*$, $\|\tau(a)-\tau(b)\|_0 > r^*$. \label{list:same-cp}
\item For all distinct $a,b$ both from $A^*$ or both from $B^*$, $\|\tau(a) - \tau(b)\|_0 \geq \Delta$. \label{list:cross-cp}
\end{enumerate}
Once again, we need an $(1 + \varepsilon)$ gap between the bounds in (\ref{list:cross-cp}) and (\ref{list:edge-cp}), i.e., $\frac{\Delta}{r^*}$. Unfortunately, we cannot construct such codes using any of the Reed-Solomon code families. We turn to another type of codes that beat the Gilbert-Varshamov bound:  Algebraic-~Geometric (AG) codes. Similar to the Reed-Solomon code based construction, we take $C^*$ as an AG code and $\tilde{C}^*$ to be a ``higher degree'' AG code; getting the desired gap simply means that the distance of $C^*$ must be at least   $(1 + \varepsilon)$ times the distance of $\tilde{C}^*$.

Recall from Subsection~\ref{subsec:abstract} also that, to bound the density of $G^*$, we need a lower bound on the number of minimum weight codewords of $\tilde{C}^*$. Such bounds for AG codes are non-trivial and we turn to the bounds from~\cite{AshikhminBV01,vluaduct2018lattices}. Unfortunately, this only gives $G^*$ with density $|C^*|^{-1/2 - o(1)}$, instead of $|C^*|^{-o(1)}$ as before. This is indeed the reason that our running time lower bound for approximate \CP is only $n^{1.5-\varepsilon}$.

\begin{sloppypar}We are not aware of any result on the (asymptotic) tightness of the bounds from~\cite{AshikhminBV01,vluaduct2018lattices} that we use. However, improving upon such bounds would have other consequences, such as a better bound on the kissing numbers of lattices constructed in~\cite{vluaduct2018lattices}. As a result, it seems likely that more understanding of AG codes (and perhaps even new constructions) are needed in order to improve these bounds.\end{sloppypar}


\section{Preliminaries}
In this section we define the geometric problems of interest to this paper, give an alternate proof for the conditional lower bound on bichromatic closest pair, and recall the definition of the contact dimension of a graph.

\subsection{Notations, Problems and Fine-Grained Hypotheses}
\paragraph{Distance Measures.}
For any two vectors $a,b\in\mathbb{R}^d$, the distance between them in the $\ell_p$-metric is denoted by $||a-b||_p~=~\left(\sum_{i=1}^d|a_i-b_i|^p\right)^{1/p}$.
Their distance in the $\ell_{\infty}$-metric is denoted by 
$||a-b||_{\infty} = \underset{{i\in[d]}}{\max}\ \{|a_i-b_i|\}$, and
in the $\ell_0$-metric is denoted by 
$||a-b||_0=|\{i\in[d]:a_i\neq b_i\}|$, i.e.,
the number of coordinates on which $a$ and $b$ differ. More generally, for any two vectors $a,b\in\mathbb{R}^d$ in the $\Delta$-metric, we denote by $\Delta(a,b)$ its distance in that metric space.
The  $\ell_p$-metrics that are well studied in literature are
the {\em Hamming metric} ($\ell_0$-metric),
the {\em rectilinear metric} ($\ell_1$-metric),
the {\em Euclidean metric} ($\ell_2$-metric), and
the {\em Chebyshev metric}  ($\ell_{\infty}$-metric). 
We denote the inner product (associated with the Euclidean space) of $a$ and $b$ by $\langle a,b\rangle=\underset{i\in[d]}{\sum}a_i\cdot b_i$. Finally, for every positive integer $d$ we define the edit metric over $\Sigma$ to be the space $\Sigma^d$ endowed with distance function $\ed(a, b)$, which is defined as the minimum number of character substitutions/insertions/deletions to transform $a$ into $b$.

\paragraph{Problems.}
Here we give formal definitions of Orthogonal Vectors (\OV), Closest Pair (\CP) and Bichromatic Closest Pair (\BCP) problems, and also Maximum Inner Product (\MIP) and Bichromatic Maximum Inner Product (\BMIP) problems.

\begin{definition}[Orthogonal Vectors Problem, \OV]
In \OV, we are given two collections of $n$ points $A,B\subseteq\mathbb \{0,1\}^d$, and the goal is to find a pair of points $a\in A$, $b\in B$ such that
$\langle a,b\rangle=0$.
\end{definition}

\begin{definition}[Closest Pair Problem, \CP]
In \CP in the $\Delta$-metric, we are given a collection of $n$ points $P\subseteq\mathbb R^d$ and a positive real $\alpha$, and the goal is to find a pair of distinct points $a,b\in P$ such that
$\Delta(a,b)\le \alpha$.
\end{definition}

\begin{definition}[Bichromatic Closest Pair Problem, \BCP]
In \BCP in the $\Delta$-metric, we are given two collections of $n$ points  $A,B\subseteq\mathbb R^d$ and a positive real $\alpha$, and the goal is to find a pair of points $a\in A$, $b\in B$ such that
$\Delta(a,b)\le \alpha$.
\end{definition}

We will also use gap versions of these problems. For any $\delta\ge 0$, we define $(1+\delta)$-\CP (resp.\ $(1+\delta)$-
\BCP) in the $\Delta$-metric to be the problem of distinguishing between the case whether there exist distinct $a,b\in P$ (resp.\ $a\in A$ and $b\in B$) such that $\Delta(a,b)\le \alpha$ and the case where for all distinct $a,b\in P$ (resp.\ $a\in A$ and $b\in B$) we have $\Delta(a,b)> (1+\delta)\cdot \alpha$. 

\begin{definition}[Maximum Inner Product Problem, \MIP]
In \MIP, we are given a collection of $n$ points $P\subseteq\mathbb R^d$ and a  real $\alpha$, and the goal is to find a pair of distinct points $a,b\in P$ such that
$\langle a,b\rangle\ge \alpha$.
\end{definition}

\begin{definition}[Bichromatic Maximum Inner Product Problem, \BMIP]
In \BMIP, we are given two collections of $n$ points $A,B\subseteq\mathbb R^d$ and a  real $\alpha$, and the goal is to find a pair of points $a\in A$, $b\in B$ such that
$\langle a,b\rangle\ge \alpha$.
\end{definition}

Again we define the gap versions of these problems as follows. For any $\gamma\ge 1$, we define $\gamma$-\MIP (resp.\ $\gamma$-\BMIP) to be the problem of distinguishing between the case whether there exist distinct $a,b\in P$ (resp.\ $a\in A$ and $b\in B$) such that $\langle a,b\rangle\ge \alpha$ and the case where for all distinct $a,b\in P$ (resp.\ $a\in A$ and $b\in B$) we have $\langle a,b\rangle<\nicefrac{\alpha}{\gamma}$.

\paragraph{Hypotheses.}
Finally, we give formal definitions of the relevant fine-grained hypotheses (see \cite{Vir18} for a survey on the state-of-the-art conditional lower bounds that are known under these hypotheses). 
\begin{definition}[Strong Exponential Time Hypothesis, \SETH \cite{IP01,IPZ01,CIP06}]
For every $\varepsilon > 0$, there exists $k = k(\varepsilon) \in \N$ such that no algorithm can solve $k$-SAT (i.e., satisfiability on a CNF of width $k$) in $O(2^{(1 - \varepsilon)m})$ time where $m$ is the number of variables. Moreover, this holds even when the number of clauses is at most $c(\varepsilon) m$ where $c(\varepsilon)$ denotes a constant that depends only on $\varepsilon$.
\end{definition}

\begin{definition}[Orthogonal Vector Hypothesis, \OVH]
For every $\varepsilon > 0$,  no algorithm can solve \OV in $O(n^{2-\varepsilon})$ time. Moreover, this holds even when the dimension $d$ is at most $c(\varepsilon) \log n$ where $c(\varepsilon)$ denotes a constant that depends only on $\varepsilon$.
\end{definition}

It is known that \SETH implies \OVH \cite{W05}, and therefore in the rest of the paper, we base all our conditional lower bounds on \OVH.

\subsection{Error-Correcting Codes} \label{sec:code-prelim}

We recall here a few coding theoretic notations since all of our gadgets are based on error-correcting codes. As is standard in error-correcting codes, we will use $\Delta(\ba, \bb)$ to denote $\|\ba - \bb\|_0$, the Hamming distance of $\ba$ and $\bb$, for any $\ba, \bb \in \mathbb{F}_q^N$ and we further define $\Delta(\ba, S) := \underset{\bb \in S}{\min}\ \Delta(\ba, \bb)$ for any $\ba \in \mathbb{F}_q^N$ and $S \subseteq \mathbb{F}_q^N$. The weight of $\ba \in \mathbb{F}_q^N$, denoted by $\Delta(\ba)$, is simply $\|\ba\|_0 := |{i \in [N] : a_i \ne 0}|$. For $\ba \in \mathbb{F}_q^N$ and $d \in \mathbb{N}$, we use $\cB(\ba, d)$ to denote the (closed) Hamming ball of radius $d$ centered at $\ba$, i.e., $\cB(\ba, d) := \{\bb \in \mathbb{F}_q^N \mid \Delta(\ba, \bb) \leq d\}$.

An error correcting code of block length $N$ over alphabet $\mathbb{F}_q$ is simply a collection of codewords $\cC \subseteq \mathbb{F}_q^N$.  The distance of the code $\cC$, denoted by $\Delta(\cC)$, is defined as $\underset{\ba \ne \bb \in \cC}{\min}\ \Delta(\ba, \bb)$. A code is said to be linear if $\cC$ is a subspace of $\mathbb{F}_q^N$. For a linear code $\cC$, its message length is defined to be the dimension of $\cC$, or equivalently $\log_q |\cC|$. We often use the notion $[N, K, D]_q$ to denote a linear code of block length $N$, message length $K$, and distance $D$. The rate and relative distance of a linear $[N, K, D]_q$ code $\cC$ are defined as $K/N$ and $D/N$ respectively. Note also that, for a linear code $\cC$, $\Delta(\cC)$ is equal to the minimum weight of a non-zero codeword of $\cC$. Finally, for any code $\cC$, we use $A_w(\cC) := |\{\bc \in \cC \mid \Delta(\bc) = w\}|$ to denote the number of codewords of weight $w$.

Let us also recall the Singleton bound and the definition of \emph{maximum distance separable} (MDS) codes.

\begin{theorem}[Singleton bound~\cite{Singleton}]
For any linear $[N, K, D]_q$ code, $K + D \leq N + 1$.
\end{theorem}

\begin{definition}[MDS Codes]
A linear $[N, K, D]_q$ code is said to be a maximum distance separable (MDS) code if it matches the Singleton bound, i.e., $K + D = N + 1$.
\end{definition}

We note here that the above bound and notation are well-defined (or can be naturally extended) also for non-linear codes, but we will only use them in context of linear codes in this paper.

\subsection{Miscellaneous Tools}

\paragraph{Covering Biclique by Isomorphic Graphs.} A useful fact we use to derandomize our reductions is that the biclique can be covered by any dense bipartite graph $G$ with only a few graphs that are isomorphic to $G$. To state this more formally, let us first define a few notions.

\begin{definition}
For any graph $G = (V_G, E_G)$ and any permutation $\pi: V_G \to V_G$, we use $G_{\pi}$ to denote the graph $(V_{G_{\pi}}, E_{G_{\pi}})$ where the vertex set $V_{G_\pi}$ is equal to $V_G$ and $E_{G_\pi} = \{(\pi(a), \pi(b)) \mid (a, b) \in E_G\}$.
\end{definition}

For brevity, we say that a permutation $\pi: A\dot\cup B \to A \dot\cup B$ of vertices of a bipartite graph $G = (A\dot\cup B, E_G)$ is \emph{side-preserving} if $\pi(A) = A$ and $\pi(B) = B$.

We can now state the result as follows. The proof, which proceeds via a simple set covering argument, is deferred to Appendix~\ref{app:biclique-cover}.

\begin{lemma}\label{lem:cover}
For any bipartite graph $G(A\dot\cup B,E_G)$ where $|A|=|B|=n$ and $E_G \ne \emptyset$, there exist side-preserving permutations $\pi_1, \dots, \pi_k: A \cup B \to A \cup B$ where $k \leq \frac{2n^2\ln n}{|E_G|} + 1$ such that $$\underset{i\in [k]}{\cup}E_{G_{\pi_i}}=E_{K_{n,n}}$$
Moreover, such permutations can be found in time $O(n^6 \log n)$.
\end{lemma}

\paragraph{Translating Finite Fields Vectors to \{0, 1\}-Vectors.} Another simple fact which was already mentioned in the proof overview (Section~\ref{sec:overview}) is that, we can embed Hamming metric on alphabet of size $q$ to Hamming metric on Boolean alphabet, with only $q$ multiplicative factor blow-up in the dimension:

\begin{proposition} \label{prop:simplex}
For any $q, N\in\mathbb{N}$, and alphabet $\Sigma$ such that $|\Sigma|=q$, there exists a mapping $\psi: \Sigma^N \to \{0, 1\}^{q \cdot N}$ such that, for all $\bv_1, \bv_2 \in \Sigma^N$, we have $\|\psi(\bv_1) - \psi(\bv_2)\|_0 = 2 \cdot \Delta(\bv_1, \bv_2)$ and $\left<\psi(\bv_1), \psi(\bv_2)\right> = N - \Delta(\bv_1, \bv_2)$.
\end{proposition}

\begin{proof}
The mapping $\psi$ simply replaces each coordinate that is equal to $j \in\Sigma$ by the $j$-th standard basis in the $q$-dimensional space. More precisely, for $\bv = (v_1, \dots, v_N) \in \mathbb{F}_q$, we define
\begin{align*}
\psi(\bv) = e_{v_1} \circ e_{v_2} \circ \cdots \circ e_{v_N},
\end{align*}
where $\circ$ denotes concatenation of vectors and $e_j$ denote the $j$-th standard basis in $\mathbb{R}^q$, i.e., the vector whose $j$-th coordinate is one and the remaining coordinates are zeroes. 

It is simple to check that this satisfies the two requirements. 
\end{proof}

\subsection{{\OVH{}}-hardness of Exact Bichromatic Closest Pair}

Alman and Williams \cite{AW15} showed the conditional hardness (under \OVH) of exact \BCP in every $\ell_p$-metric  even when the point-sets are over $\{0,1\}$ via a Turing reduction from \OV. David, Karthik, and Laekhanukit \ \cite{DKL18} gave an alternate proof of the same result where point-sets were over $\mathbb{R}$ via a many-one reduction from \OV.
For independent interest, below we give another proof, which is both a many-one reduction and the point-sets are over $\{0,1\}$.

\begin{theorem}\label{thm:BCP}
Assuming \OVH, for every $\varepsilon>0$, no algorithm running in time $n^{2-\varepsilon}$  can solve \BCP, even when the point-sets $A,B$ are subsets of $\{0,1\}^d$ and $d=c_\varepsilon\log n$, for some constant $c_\varepsilon>1$ (only depending on $\varepsilon$).
\end{theorem}
\begin{proof}
Let $A,B\subseteq \{0,1\}^d$ where $|A|=|B|=n$ be the input to an \OV instance. We build an instance $(A',B',\alpha)$ of \BCP where $A',B'\subseteq\{0,1\}^{5d}$, $|A|=|B|=n$, and $\alpha=2d$, using functions $T_A$ and $T_B$ guaranteed by the following claim.
\begin{claim}\label{claim:T}
There are functions $T_A,T_B:\{0,1\}\to\{0,1\}^5$ such that for every $x,y\in\{0,1\}$ we have:
\begin{itemize}
\item $x\cdot y=0$ implies $\|T_A(x)-T_B(y)\|_0=2$.
\item $x\cdot y=1$ implies $\|T_A(x)-T_B(y)\|_0=4$.
\end{itemize}
\end{claim}

For every $i\in[n]$, the $i^{\text{th}}$ point of $A'$, say $a'$ is constructed from the $i^{\text{th}}$ point of $A$, say $a$ by simply applying $T_A$ pointwise on each coordinate of $a$, i.e., $a'=(T_A(a_1),\ldots ,T_A(a_d))$. Similarly we apply $T_B$ pointwise on each coordinate of points in $B$. It is easy to see that there exists $(a_i',b_j')\in A'\times B'$ such that $\|a_i'-b_j'\|_0=2d$ if and only if $\langle a_i,b_j\rangle=0$, and otherwise every pair of points in $A'\times B'$ is at Hamming distance at least $2d+2$.
\end{proof}

\begin{proof}[Proof of Claim~\ref{claim:T}]
We define for all $x,y\in\{0,1\}$, $T_A(x)=(T_A(x)_{0,0},T_A(x)_{0,1},T_A(x)_{1,0},x,0)$ and $T_B(y)=(T_B(y)_{0,0},T_B(y)_{0,1},T_B(y)_{1,0},0,y)$, where for all $i,j\in\{0,1\}$ such that $i\cdot j=0$, we have $T_A(x)_{i,j}=1$ if and only if $x=i$ and $T_B(y)_{i,j}=1$ if and only if $y=j$. More succinctly, $T_A$ and $T_B$ are described below as strings and the claim follows by a straight-forward calculation. 
$$
T_A(0)=11000\ \ \  \ \ \ \ T_A(1)=00110$$
$$
T_B(0)=10100\ \ \  \ \ \ \ T_B(1)=01001 \qedhere
$$
\end{proof}

\subsection{Contact Dimension of a Graph}

The central gadget in our reduction from \BCP to \CP  is based on the contact dimension of a graph. Below we reproduce its definition from the proof overview (i.e. Definition~\ref{def:cd-informal}) for convenience.

\begin{definition}[Contact Dimension~\cite{P80}] \label{def:cd}
For any graph $G = (V, E)$, a mapping $\tau: V \to \mathbb{R}^d$ is said to \emph{realize $G$ (in the $\ell_p$-metric)} if for some $\beta > 0$, the following holds:
\begin{enumerate}[(i)]
\item For all $(u, v) \in E$, $\|\tau(u) - \tau(v)\|_p = \beta$.
\item For all $(u, v) \notin E$, $\|\tau(u) - \tau(v)\|_p > \beta$.
\end{enumerate}
The \emph{contact dimension} (in the $\ell_p$-metric) of $G$, denoted by $\cd_p(G)$, is the minimum $d \in \mathbb{N}$ such that there exists $\tau: V \to \mathbb{R}^d$ realizing $G$ in the $\ell_p$-metric.
\end{definition}

We may also say that $\tau$ \emph{$\beta$-realizes} $G$ if we wishes to emphasize the value of $\beta$. 

Note here that we may view points in $\tau(V)$ as centers of spheres of radius $\beta/2$. No two spheres overlap but they may touch, and $G$ has an edge $(u,v)$ if and only if the spheres centered at $\tau(u)$ and $\tau(v)$ touches.

For a summary of the bounds on $\cd(G)$ for various graphs in the Euclidean metric see \cite{M85,FM86,FM88,M91} and for a summary of the bounds on $\cd(K_{n,n})$ in various metrics see  \cite{DKL18}. For this paper, the following bounds are relevant. 

\begin{theorem}[Frankl-Maehara \cite{FM88}]
$(1.286)n -1 < \cd_2(K_{n,n}) < (1.5)n. $
\end{theorem}

\begin{theorem}[David-Karthik-Laekhanukit \cite{DKL18}]\label{thm:L0DKL18}
$\cd_0(K_{n,n})=n.$
\end{theorem}

In particular, the above two theorems are the obstacles of the approach of~\cite{DKL18} for the $\ell_2$ and Hamming metrics respectively. As discussed in the proof overview, we will overcome these barriers by constructing dense bipartite graphs with low contact dimensions in every $\ell_p$ metrics.

As discussed in Section~\ref{sec:apx-cp-overview}, we need a generalization of contact dimension in order to show inapproximability for \CP. This is formally defined below; it should be noted that the definition only makes sense for bipartite graphs, whereas the original contact dimension is well-defined for any graphs. Moreover, when $\lambda = 1$, the notion of gap contact dimension coincides with the (non-gap) contact dimension in bipartite graphs.

\begin{definition}[Gap Contact Dimension]
For any bipartite graph $G = (A\dot\cup B, E)$ and $\lambda \geq 1$, a mapping $\tau: V \to \mathbb{R}^d$ is said to \emph{$\lambda$-gap-realize $G$ (in the $\ell_p$-metric)} if for some $\beta > 0$, the following holds:
\begin{enumerate}[(i)]
\item For all $(u, v) \in E$, $\|\tau(u) - \tau(v)\|_p = \beta$.
\item For all $(u, v) \in (A \times B) \setminus E$, $\|\tau(u) - \tau(v)\|_p > \beta$.
\item For all distinct $u, v$ both from $A$ or both from $B$, $\|\tau(u) - \tau(v)\|_p > \lambda \cdot \beta$.
\end{enumerate}
The \emph{$\lambda$-gap contact dimension} (in the $\ell_p$-metric) of $G$, denoted by $\lambda\text{-}\cd_p(G)$, is the minimum $d \in \mathbb{N}$ such that there exists $\tau: V \to \mathbb{R}^d$ $\lambda$-gap-realizing $G$ in the $\ell_p$-metric.
\end{definition}

Again, we may say that $\tau$ $(\beta, \lambda)$-gap-realizes $G$ to emphasize the value of $\beta$.

Finally, we define an analogous notion for inner product:
\begin{definition}[Gap Inner Product Dimension]
For any bipartite graph $G = (A\dot\cup B, E)$ and $\lambda \geq 1$, a mapping $\tau: V \to \mathbb{R}^d$ is said to \emph{$\lambda$-gap-\IP-realize $G$} if for some $\beta > 0$, the following holds:
\begin{enumerate}[(i)]
\item For all $(u, v) \in E$, $\langle\tau(u) , \tau(v)\rangle = \beta$.
\item For all $(u, v) \in (A \times B) \setminus E$, $\langle\tau(u) , \tau(v)\rangle < \beta$.
\item For all distinct $u, v$ both from $A$ or both from $B$, $\langle\tau(u) , \tau(v)\rangle < \beta/\lambda$.
\end{enumerate}
The \emph{$\lambda$-gap inner product dimension} of $G$, denoted by $\lambda\text{-}\ipd(G)$, is the minimum $d \in \mathbb{N}$ such that there exists $\tau: V \to \mathbb{R}^d$ $\lambda$-gap-\IP-realizing $G$.
\end{definition}

We may say that $\tau$ $(\beta, \lambda)$-gap-\IP-realizes $G$ to emphasize the value of $\beta$.

\section{Lower Bound on Closest Pair under Orthogonal Vector Hypothesis}\label{sec:CP}

In this section, we prove the subquadratic hardness for \CP (assuming \OVH) using the efficient construction of a realization of a dense bipartite graph. The construction will be be formally stated below and the details will be given in Section~\ref{sec:gadet-cp}. First, we define the notion of a \emph{log-dense} sequence of integers:
\begin{definition}
A sequence $(n_i)_{i \in \mathbb{N}}$ of increasing positive integers is said to be \emph{log-dense} if there exists a constant $C \geq 1$ such that $\log n_{i + 1} \leq C \cdot \log n_i$ for all $i \in \mathbb{N}$. 
\end{definition}

As outlined in Section~\ref{sec:overviewRS} , we use Reed-Solomon codes to construct a family of dense bipartite graphs with low contact dimensions. While the construction does not yield a graph for every number of vertices $n$, it does yield a graph for a log-dense sequence of numbers of vertices, which turns out to be sufficient for the purpose of the reduction. More formally, we will prove the following in Section~\ref{sec:gadet-cp}.


\begin{theorem} \label{thm:dense-cd}
For every $0 < \delta < 1$, there exists a log-dense sequence $(n_i)_{i \in \mathbb{N}}$ such that, for every $i \in \mathbb{N}$, there is a bipartite graph $G_i = (A_i \dot \cup B_i,E_i)$ where $|A_i| = |B_i| = n_i$ and $|E_i| \geq \Omega(n_i^{2 - \delta})$, such that $\cd(G_i) = (\log n_i)^{O(1/\delta)}$. Moreover, for all $i \in \mathbb{N}$, a realization $\tau: A_i \dot \cup B_i \to \{0, 1\}^{(\log n_i)^{O(1/\delta)}}$ of $G_i$ can be constructed in time $n_i^{2 + o(1)}$.
\end{theorem}

Notice that we did not specify any $\ell_p$-metric in the notion of contact dimension above. This is intentional, because our point sets $\tau(A_i \dot \cup B_i)$ have coordinate entries in $\{0, 1\}$, for which the distances in the Hamming metric are equivalent (up to power of $p$) to distances in any $\ell_p$-metric ($p\neq \infty$). We also adopt this notational convenience below. Specifically, we will prove the following theorem which states that \CP is hard even when the points are from $\{0, 1\}^d$; clearly, this also implies Theorem~\ref{thm:CP} due to the aforementioned equivalence to other $\ell_p$-metrics.

\begin{theorem}[Subquadratic Hardness of $\{0, 1\}$-\CP]\label{thm:CP-01}
Assuming \OVH, for every $\varepsilon > 0$, there exists $s_{\varepsilon} > 0$ such that no algorithm running in $O(n^{2-\varepsilon})$  time can solve \CP in the Hamming metric even when $d=\left(\log n\right)^{s_\varepsilon}$ and all points have $\{0, 1\}$ entries.
\end{theorem}


\begin{proof}
For any $\varepsilon > 0$, let $C_{\text{exp}}$ be the constant such that the dimension guarantee for $\tau$ in Theorem~\ref{thm:dense-cd} is at most $(\log n_i)^{C_{\text{exp}}/\varepsilon}$ for $\delta = \varepsilon/2$. We define $s_\varepsilon$ as $2 \cdot C_{\text{exp}} / \varepsilon + 2$.

Assume  that there exists $\varepsilon > 0$ and an algorithm $\cA$ that can solve \CP in time $n^{2 - \varepsilon}$ in the Hamming metric for any input of $n$ points in   $\{0,1\}^{(\log n)^{s_{\varepsilon}}}$. We will construct an algorithm $\cA'$ that solves any instance of \BCP in time $n^{2 - \varepsilon'}$ for some constant $\varepsilon' > 0$ (to be specified below), on $n$ points in  dimension $d := c_{\varepsilon'} \cdot \log n$ with coordinate entries in $\{0,1\}$. Together with Theorem~\ref{thm:BCP}, this implies that \OVH is false, arriving at a contradiction.

Let $C_\varepsilon$ denote the log-density constant (i.e. $\sup_{i} \frac{\log n_{i+1}}{\log n_i}$) of the sequence from Theorem~\ref{thm:dense-cd} for $\delta = \varepsilon/2$, and let $\varepsilon'$ be $0.01 \cdot \varepsilon/ C_{\varepsilon}$. The algorithm $\cA'$ on input $(A, B, \alpha)$ where $A, B \subseteq \{0, 1\}^d,$ with $|A|=|B|=n$, and $  \alpha \in [d]$, works as follows:
\begin{enumerate}
\item Let $n'$ be the largest number in the sequence from Theorem~\ref{thm:dense-cd} with $\delta = \varepsilon/2$ s.t. $n' \leq n^{0.1}$.
\item Let $G' = (A' \dot\cup B', E')$ be the graph from Theorem~\ref{thm:dense-cd} with $|A'| = |B'| = n'$, $|E'| \geq \Omega((n')^{2-\delta})$, and $\tau: A'\dot\cup B' \to \{0, 1\}^{(\log n')^{C_{\text{exp}}/\varepsilon}}$ be a $\beta$-realization of $G'$ where $\beta \in \mathbb{N}$.
\item We use the algorithm from Lemma~\ref{lem:cover} to find $\pi_1, \dots, \pi_k$ where $k = O((n')^{\delta}\log n')$ such that the union of $E_{G'_{\pi_1}}, \ldots, E_{G'_{\pi_k}}$ is $E_{K_{n', n'}}$.
\item We assume w.l.o.g.\footnote{This is without loss of generality, since if $n$ is not divisible by $n'$, we can use brute force for the remainder points. This requires only $O(n \cdot n' \cdot) = O(n^{1.1} \log n)$ which does not affect the overall asymptotic running time of the algorithm.} that $n$ is divisible by $n'$. Partition $A$ and $B$ into $A_1, \dots, A_{n/n'}$ and $B_1, \dots, B_{n/n'}$ each of size $n'$. For each $i, j \in [n/n'], t \in [k]$, do the following:
\begin{enumerate}
\item 
Let $\tau_t$ be an appropriate permutation of $\tau$ that $\beta$-realizes $G'_{\pi_t}$. Label the vertices of $G'_{\pi_t}$ with the points in $A_i\dot\cup B_j$.
\item Let $\alpha' = \alpha + (d + 1) \cdot \beta$, and define $A_i^t, B_j^t$ as
\begin{align*}
A_i^t = \{\ba \circ (\bone_{d + 1} \otimes \tau_t(\ba)) \mid \ba \in A_i\}, B_j^t = \{\bb \circ (\bone_{d + 1} \otimes \tau_t(\bb)) \mid \bb \in B_j\}
\end{align*}
where $\bone_{d+1} \otimes \bv$ simply denotes $\bv \circ \bv \circ \cdots \circ \bv$, i.e., the concatenation of $d + 1$ copies of $\bv$.
\item Run $\cA$ on $(A_i^t \dot\cup B_j^t, \alpha')$. If $\cA$ outputs YES, then output YES and terminate.
\end{enumerate}
\item If none of the executions of $\cA$ returns YES, then output NO.
\end{enumerate}

Observe that the bottleneck in the running time of the algorithm is in the executions of $\cA$. The number of executions is $(n/n')^2 \cdot k$ and each execution takes $O((n')^{2 - \varepsilon})$ time. Hence, in total the running time of the algorithm $\cA'$ is $O((n/n')^2 \cdot k \cdot (n')^{2 - \varepsilon}) \leq O(n^2 \log n \cdot (n')^{-\varepsilon/2})$. Now, from the log-density of the sequence from Theorem~\ref{thm:dense-cd}, we have $n' \geq n^{0.1/C_\varepsilon} = n^{10\varepsilon'/\varepsilon}$. As a result, the running time of $\cA$ is at most $O(n^{2 - 5\varepsilon'}\log n) \leq O(n^{2 - \varepsilon'})$ as desired.

To see the correctness of the algorithm, first observe that the dimensions of vectors in $A_i^t, B_j^t$ are at most $d + (d + 1) \cdot (\log n')^{C_{\text{exp}}/\varepsilon}$ which is at most $(\log {n})^{s_{\varepsilon}}$ for any sufficiently large $n$; that is, the calls to $\cA$ are valid. Next, observe that, if $(A, B, \alpha)$ is a YES instance of \BCP, there must be $i, j \in [n/n']$ and $\ba^* \in A_i, \bb^* \in B_j$ such that $\|\ba^* - \bb^*\|_0$ is at most $\alpha$. Since $G'_{\pi_1}, \dots, G'_{\pi_k}$ covers $K_{n', n'}$, there must be $t \in [k]$ such that $\|\tau_t(\ba^*) - \tau_t(\bb^*)\|_0 =\beta$. As a result, $\|(\ba^* \circ (\bone_{d + 1} \otimes \tau_t(\ba^*))) - (\bb^* \circ (\bone_{d + 1} \otimes \tau_t(\bb^*)))\|_0 \leq \alpha + (d + 1) \cdot \beta = \alpha'$. Thus, $(A_i^t \cup B_j^t, \alpha')$ is a YES instance for \CP and $\cA'$ outputs YES as desired.

Finally, assume that $(A, B, \alpha)$ is a NO instance of \BCP. Consider any $i, j \in [n/n']$ and $t \in [k]$. To argue that $(A_i^t \cup B_j^t, \alpha')$ is a NO instance for \CP, we have to show that any two points in $A_i^t \cup B_j^t$ have distance more than $\alpha'$. To see this, let us consider two cases.
\begin{enumerate}
\item Both points are either from $A_i^t$ or from $B_j^t$. Assume w.l.o.g. that the two points are from $A_i^t$; let them be $\ba \circ (\bone_{d + 1} \otimes \tau_t(\ba))$ and $\ba' \circ (\bone_{d + 1} \otimes \tau_t(\ba'))$. Recall that, from the definition of $\beta$-realization, $\|\tau_t(\ba) - \tau_t(\ba')\|_0 > \beta$. Since $\|\tau_t(\ba) - \tau_t(\ba')\|_0$ is an integer, we must have $\|\tau_t(\ba) - \tau_t(\ba')\|_0 \geq \beta + 1$. As a result, the Hamming distance between the two points is at least $(d + 1) \cdot (\beta + 1) > d + (d + 1) \cdot \beta = \alpha'$.
\item One of the point is from $A_i^t$ and the other from $B_j^t$. Let them be $\ba \circ (\bone_{d + 1} \otimes \tau_t(\ba))$ and $\bb \circ (\bone_{d + 1} \otimes \tau_t(\bb))$. Since $(A, B, \alpha)$ is a NO instance of \BCP, $\|\ba - \bb\|_0 > \alpha$. Furthermore, from definition of $\beta$-realization, we must have $\|\tau_t(\ba) - \tau_t(\bb)\|_0 \geq \beta$. Combining the two implies that the Hamming distance between $\ba \circ (\bone_{d + 1} \otimes \tau_t(\ba))$ and $\bb \circ (\bone_{d + 1} \otimes \tau_t(\bb))$ is more than $\alpha'$. 
\end{enumerate}
Hence, $(A_i^t \dot\cup B_j^t, \alpha')$ must be a NO instance for \CP for every $t\in[k]$ and $i,j\in [n/n']$. Thus, $\cA'$ outputs NO as desired.
\end{proof}


\section{Gadget Constructions}\label{sec:cover}

In this section, we construct all the gadgets that are used in our reductions, including the basic gadget (Theorem~\ref{thm:dense-cd}) and more advanced gadgets used for \MIP and approximate version of \CP.

\subsection{Finding a Center of a Code via Another Code}

At the heart of all our gadgets is the task of finding  a code $\cC_1$ and a center $\bs$ such that there are  $|\cC_1|^{1-o(1)}$ many codewords at  Hamming distance exactly equal to $r$ (for some $r>0$) from $\bs$ but there is no codeword in $\cC_1$ at  distance less than $r$ from $\bs$. The below lemma is  useful in finding such an $\bs$.

\begin{lemma} \label{lem:finding-center}
Let $\cC_1 \subseteq \cC_2 \subseteq \mathbb{F}_q^N$ be two linear codes with the same block length $N$ and alphabet $\mathbb{F}_q$ such that $\Delta(\cC_2) < \Delta(\cC_1)$. Then, there exists a center $\bs \in \mathbb{F}_q^N$ such that (1) $\Delta(\bs, \cC_1) \geq \Delta(\cC_2)$ and (2) $|\cB(\bs, \Delta(\cC_2)) \cap \cC_1| / |\cC_1| \geq A_{\Delta(\cC_2)}(\cC_2) / |\cC_2|$. Moreover, given $\cC_1, \cC_2$, such an $\bs$ can be found in $O(|\cC_1| \cdot |\cC_2| \cdot qN)$ time.
\end{lemma}

\begin{proof}
We  show that there exists $\bs \in \cC_2 \setminus \cC_1$ such that (2) holds. Note that (1) immediately holds, because $\bs - \bc$ must be a non-zero codeword of $\cC_2$ which implies that $\Delta(\bs, \bc) \geq \Delta(\cC_2)$. 

\begin{sloppypar}To show that there exists $\bs \in \cC_2 \setminus \cC_1$ such that $|\cB(\bs, \Delta(\cC_2)) \cap \cC_1| \geq |\cC_1| \cdot A_{\Delta(\cC_2)} / |\cC_2|$. We will in fact show a stronger statement: for a random $\bs \in \cC_2 \setminus \cC_1$, we have $\E[|\cB(\bs, \Delta(\cC_2)) \cap \cC_1|] \geq |\cC_1| \cdot A_{\Delta(\cC_2)} / |\cC_2|$. Consider $\E_{\bs \in \cC_2 \setminus \cC_1}[|\cB(\bs, \Delta(\cC_2)) \cap \cC_1|]$. Due to linearity of expectation, we have\end{sloppypar}
\begin{align*}
\E_{\bs \in \cC_2 \setminus \cC_1}[|\cB(\bs, \Delta(\cC_2)) \cap \cC_1|] 
&= \sum_{\bc \in \cC_1} \Pr_{\bs \in \cC_2 \setminus \cC_1}[\bc \in \cB(\bs, \Delta(\cC_2))] \\
&= \sum_{\bc \in \cC_1} \Pr_{\bs \in \cC_2 \setminus \cC_1}[\Delta(\bs - \bc) \leq \Delta(\cC_2)] \\
&= \sum_{\bc \in \cC_1} \Pr_{\bs \in \cC_2 \setminus \cC_1}[\Delta(\bs) \leq \Delta(\cC_2)] \\
&= |\cC_1| \cdot \frac{|(\cC_2 \setminus \cC_1) \cap \cB(\bzero, \Delta(\cC_2))|}{|\cC_2 \setminus \cC_1|}.
\end{align*}
Now, since $\Delta(\cC_1) > \Delta(\cC_2)$, we have $\cC_1 \cap \cB(\bzero, \Delta(\cC_2)) = \{\bzero\}$. That is, $|(\cC_2 \setminus \cC_1) \cap \cB(\bzero, \Delta(\cC_2))| = |(\cC_2 \setminus \{\bzero\}) \cap \cB(\bzero, \Delta(\cC_2))| = A_{\Delta(\cC_2)}(\cC_2)$. Plugging this back into the above equality, we have
\begin{align*}
\E_{\bs \in \cC_2 \setminus \cC_1}[|\cB(\bs, \Delta(\cC_2)) \cap \cC_1|] &= |\cC_1| \cdot \frac{A_{\Delta(\cC_2)}(\cC_2)}{|\cC_2 \setminus \cC_1|} \geq |\cC_1| \cdot \frac{A_{\Delta(\cC_2)}(\cC_2)}{|\cC_2|}.
\end{align*}
Thus, there must exist a center $\bs \in \cC_2 \setminus \cC_1$ that satisfies (2) (and also (1)) as desired.

Finally, note that $\bs$ can be found by a brute force algorithm that tries every $\bs \in \cC_2$ and check whether (2) is satisfied; this algorithm takes $O(|\cC_1| \cdot |\cC_2| \cdot qN)$ time.
\end{proof}

\subsection{Gadgets based on Reed-Solomon Codes}

In this subsection, we construct gadgets based on the Reed Solomon codes, which are defined below.

\begin{theorem}[Reed-Solomon Codes]
For every prime power $q$, and every $K \leq N \leq q$, there exists a $[N, K, N - K + 1]_q$ linear code, denoted by $\rs_q[N, K]$. The generator matrix of this code can be computed in time $\poly(N, K, q)$. Moreover, for every $q \geq N \geq K_2 > K_1$, we have $\rs_q[N, K_1] \subseteq \rs_q[N, K_2]$.
\end{theorem}

In order to  find a good center $\bs$, we use the following (well-known) bound on the number of minimum weight codewords of Reed Solomon codes (and more generally MDS codes). For a reference of this bound, see e.g.~\cite[Ch. 11, Theorem 6]{MS77}. 

\begin{lemma} \label{lem:mds}
Let $\cC$ be any linear $[N, K, D]_q$ code that is MDS. Then, $A_{D}(\cC) = \binom{N}{K - 1} \cdot (q - 1)$.
\end{lemma}

\subsubsection{The Basic Gadget: Dense Bipartite Graphs with Low Contact Dimensions} \label{sec:gadet-cp}

Now we construct a dense bipartite graph with low contact dimension. A proof sketch of this construction was provided in Section~\ref{sec:overviewRS} and was formally stated as Theorem~\ref{thm:dense-cd}.

\begin{proof}[Proof of Theorem~\ref{thm:dense-cd}]
Let $q_i$ be the $i$-th prime number and let $n_i = (q_i)^{(\lfloor q_i^{\delta}\rfloor)}$; it is simple to see that the sequence $(n_i)_{i \in \mathbb{N}}$ is log-dense. For $q = q_i$, consider the Reed-Solomon codes $\cC_1 = \rs_q[q, K_1]$ and $\cC_2 = \rs_q[q, K_2]$ where $K_1 = \lfloor q^{\delta} \rfloor$ and $K_2 = K_1 + 1$. Applying Lemma~\ref{lem:finding-center} with $(\cC_1, \cC_2)$ implies that there exists a center $\bs \in \cC_2$ such that
\begin{align*}
\frac{|\cB(\bs, \Delta(\cC_2)) \cap \cC_1|}{|\cC_1|} 
&\geq \frac{A_{\Delta(\cC_2)}}{|\cC_2|} \\
(\text{By Lemma}~\ref{lem:mds}) &= \frac{\binom{q}{K_2 - 1} \cdot (q - 1)}{q^{K_2}} \\
&\geq \frac{\left(\frac{q}{K_2 - 1}\right)^{K_2 - 1} \cdot (q - 1)}{q^{K_2}} \\
&= \frac{q - 1}{q} \cdot \left(\frac{1}{K_2 - 1}\right)^{K_2 - 1} \\
&= \frac{q - 1}{q} \cdot \frac{1}{K_1^{K_1}} \\
&\geq \frac{1}{2} \cdot \frac{1}{q^{\delta K_1}} \\
&= \Omega(|\cC_1|^{-\delta}),
\end{align*}
where the last equality follows from the fact that $|\cC_1| = q^{K_1}$.

We construct the graph $G_i = (A_i, B_i, E_i)$ and a realization $\tau$ as follows. Let $A_i = \cC_1, B_i = \{\bs + \bc \mid \bc \in \cC_1\}$ and $E_i = \{(\ba, \bb) \in A_i \times B_i \mid \Delta(\ba, \bb) = \Delta(\cC_2)\}$. $G_i$ can be easily realized by applying the mapping $\psi: \mathbb{F}_q^q \to \{0, 1\}^{q^2}$ from Proposition~\ref{prop:simplex}. More precisely, let $\tau$ be the restriction of $\psi$ on $A_i \cup B_i$. Below we argue about the density of $G_i$ and that $\tau$ is a $2\Delta(\cC_2)$-realization of $G_i$.

\begin{itemize}
\item First, notice that $|E_i|$ is exactly $|\cC_1| \cdot |\cB(\bs, \Delta(\cC_2)) \cap \cC_1| \geq \Omega(|\cC_1|^{2 - \delta}) = \Omega(n_i^{2 - \delta})$. 
\item Second, notice that, for every $\bv_1, \bv_2$ both from $A_i$ or both from $B_i$, we have $\bv_1 - \bv_2 \in \cC_1 \setminus \{\bzero\}$. This implies that $\|\tau(\bv_1) - \tau(\bv_2)\|_0 = 2\Delta(\bv_1, \bv_2) \geq 2\Delta(\cC_1) > 2\Delta(\cC_2)$.
\item Third, for every $\ba \in A_i$ and $\bb \in B_i$, we have $\ba - \bb \in \cC_2 \setminus \{\bzero\}$. Thus, $\Delta(\ba, \bb) \geq \Delta(\cC_2)$. Hence, $\|\tau(\ba) - \tau(\bb)\|_0 = 2\Delta(\ba, \bb) \geq 2\Delta(\cC_2)$. Moreover, the inequality is an equality if and only if $\Delta(\ba, \bb) = \Delta(\cC_2)$, i.e., $(\ba, \bb) \in E_i$ as desired. 
\item Finally, observe that the dimension is $q^2 = (\log n_i)^{O(1/\delta)}$.
\end{itemize}

As for the running time of constructing $G_i$ and $\tau$, observe that the bottleneck is the running time needed to find the center $\bs$; according to Lemma~\ref{lem:finding-center}, $\bs$ can be computed in $O(|\cC_1| \cdot |\cC_2| \cdot q^2) = O(n_i^2 \cdot q^2)$, which is $n_i^{2 + o(1)}$ as desired.
\end{proof}

\subsubsection{A Gadget for Maximum Inner Product}

Now, we build gadgets (stated below) which will be used for proving the inapproximability of \MIP.

\begin{theorem} \label{thm:mip-pre-tensor}
For every $0 < \delta < 1$, there exists a log-dense sequence $(n_i)_{i \in \mathbb{N}}$ such that, for every $i \in \mathbb{N}$, there is a bipartite graph $G_i = (A_i \dot \cup B_i,E_i)$ where $|A_i| = |B_i| = n_i$ and $|E_i| \geq \Omega(n_i^{2 - \delta})$, such that 3-$\ipd(G) = (\log n_i)^{O(1/\delta)}$. Moreover, for all $i \in \mathbb{N}$, a 3-gap-\IP-realization $\tau: A_i \dot \cup B_i \to \{0, 1\}^{(\log n_i)^{O(1/\delta)}}$ of $G_i$ can be constructed in time $n_i^{4 + o(1)}$.
\end{theorem}

\begin{proof}
The proof here is exactly the same as the proof of Theorem~\ref{thm:dense-cd}, except that we will not pick $K_2 = K_1 + 1$, but rather pick $K_2 > 3K_1$ (and $n_i$ accordingly). 

More precisely, let $q_i$ be the $i$-th prime number and let $n_i = (q_i)^{(\lfloor q_i^{0.3\delta}/3\rfloor)}$; it is simple to see that the sequence $(n_i)_{i \in \mathbb{N}}$ is log-dense. For $q = q_i$, consider the Reed-Solomon codes $\cC_1 = \rs_q[q, K_1]$ and $\cC_2 = \rs_q[q, K_2]$ where $K_1 = \lfloor q^{0.3\delta}/3 \rfloor$ and $K_2 = 3K_1 + 1$. Similar to the proof of Theorem~\ref{thm:dense-cd}, applying Lemma~\ref{lem:finding-center} with $(\cC_1, \cC_2)$ implies that there exists $\bs \in \cC_2 \setminus \cC_1$ such that
\begin{align*}
\frac{|\cB(\bs, \Delta(\cC_2)) \cap \cC_1|}{|\cC_1|} 
\geq \frac{q - 1}{q} \cdot \left(\frac{1}{K_2 - 1}\right)^{K_2 - 1}
= \frac{q - 1}{q} \cdot \frac{1}{(3K_1)^{(3K_1)}}
\geq \frac{1}{2} \cdot \frac{1}{q^{\delta K_1}}
= \Omega(|\cC_1|^{-\delta}).
\end{align*}

We construct the graph $G_i = (A_i, B_i, E_i)$ and a realization $\tau$ as follows. Let $A_i = \cC_1, B_i = \{\bs + \bc \mid \bc \in \cC_1\}$ and $E_i = \{(\ba, \bb) \in A_i \times B_i \mid \Delta(\ba, \bb) = \Delta(\cC_2)\}$. $G_i$ can be easily 3-gap-\IP-realized by applying the mapping $\psi: \mathbb{F}_q^q \to \{0, 1\}^{q^2}$ from Proposition~\ref{prop:simplex}. More precisely, let $\tau$ be the restriction of $\psi$ on $A_i \cup B_i$. Below we argue about the density of $G_i$ and that $\tau$ is a $(K_2 - 1, 3)$-gap-\IP-realization of $G_i$.


\begin{itemize}
\item First, notice that $|E_i|$ is exactly $|\cC_1| \cdot |\cB(\bs, \Delta(\cC_2)) \cap \cC_1| \geq \Omega(|\cC_1|^{2 - \delta}) = \Omega(n_i^{2 - \delta})$. 
\item Second, for every $\bv_1, \bv_2$ both from $A_i$ or both from $B_i$, we have $\bv_1 - \bv_2 \in \cC_1 \setminus \{\bzero\}$. Thus, $\left<\tau(\bv_1), \tau(\bv_2)\right> = q - \Delta(\bv_1, \bv_2) \leq q - \Delta(\cC_1) = K_1 - 1 < (K_2 - 1)/3$.
\item Third, for every $\ba \in A_i$ and $\bb \in B_i$, we have $\ba - \bb \in \cC_2 \setminus \{\bzero\}$. Thus, $\Delta(\ba, \bb) \geq \Delta(\cC_2)$. Hence, $\left<\tau(\ba), \tau(\bb)\right> = q - \Delta(\ba, \bb) \leq q - \Delta(\cC_2) = K_2 - 1$. Moreover, the inequality is an equality if and only if $\Delta(\ba, \bb) = \Delta(\cC_2)$, i.e., $(\ba, \bb) \in E_i$ as desired. 
\item Finally, observe that the dimension is $q^2 = (\log n_i)^{O(1/\delta)}$.
\end{itemize}
Once again, the running time of the construction is $O(|\cC_1| \cdot |\cC_2| \cdot q^2) \leq n_i^{4 + o(1)}$.
\end{proof}

\subsection{Gadgets based on AG Codes}

In this subsection, we  construct gadgets based on algebraic geometric (AG) codes. The definitions of AG Codes are well beyond the scope of this work and we refer the readers to \cite{Stichtenoth08,VNT07} for more thorough introductions.

Once again to find a good center, we need a bound on the number of minimum weight codewords. On this front, we use the following bound\footnote{Note that most of the proof of this bound was from~\cite{AshikhminBV01}; \cite{vluaduct2018lattices} simply makes the bound more explicit, which is more convenience for us.} from~\cite{vluaduct2018lattices}.  Throughout this subsection, we follow the notations from~\cite{vluaduct2018lattices}.

\begin{theorem}[Theorem 4.3 of~\cite{vluaduct2018lattices}] \label{thm:vladut}
Let $q$ be a prime power, $X$ be a curve of genus $g$ over $\mathbb{F}_q$, let $S \subseteq X(\mathbb{F}_q)$ such that $|S| = N$, and let $a \in \mathbb{N}$ with $1 \leq a \leq N - 1$. Then, there exists an $\mathbb{F}_q$-positive divisor $D \geq 0$, $\deg(D) = a$, such that the corresponding AG Code $\cC = \cC(X, D, S)$ has  minimum distance $N - a$ and
\begin{align*}
A_{N - a}(\cC) \geq \frac{\binom{N}{a}}{(\sqrt{q} + 1)^{2g}}.
\end{align*}
\end{theorem}

We also need the following well-known (central) fact about the parameters of AG codes.

\begin{theorem} \label{thm:ag-basic}
Let $q$ be a prime power, $X$ be a curve of genus $g$ over $\mathbb{F}_q$, let $S \subseteq X(\mathbb{F}_q)$ such that $|S| = N$, and let $a \in \mathbb{N}$ with $1 \leq a \leq N - 1$. Then, the corresponding AG Code $C = C(X, D, S)$ is a linear code over $\mathbb{F}_q$ with block length $N$, distance at least $N - a$ and message length $k \geq a - g + 1$.
\end{theorem}

Recall also the tower of functions of Garcia and Stichtenoth~\cite{GS96}, whose parameters approach the TVZ bound. We note here that,  it suffices for us to have the genus approaching $\Omega(N/\sqrt{q})$ and there are also other curves that satisfy this.

\begin{theorem}[\cite{GS96}] \label{thm:gs}
For any $\zeta > 0$ and any square of prime $q$, there exists a dense sequence\footnote{A sequence $(N_i)_{i \in \mathbb{N}}$ of increasing positive integers is said to be \emph{dense} if there exists a constant $C \geq 1$ such that $N_{i + 1} \leq C \cdot N_i$ for all $i \in \mathbb{N}$. 
} $(N_i)_{i \in \mathbb{N}}$ such that there exists a curve $X_i$ with genus at most $\frac{N_i}{\sqrt{q} - 1} + \zeta$ where $|X_i(\mathbb{F}_q)| \geq N_i$.
\end{theorem}

Plugging the bound from~\cite{vluaduct2018lattices} into the above family of curves immediately yields the following:

\begin{lemma} \label{lem:ag-pair}
For any $\zeta > 0$ and any square of prime $q$, there exists a dense sequence $(N_i)_{i \in \mathbb{N}}$ such that the following holds. For any $i \in \mathbb{N}$ and any $a_1, a_2 \in \mathbb{N}$ such that $1 \leq a_1 < a_2 \leq N_i - 1$, there exists linear codes $\cC_1 \subseteq \cC_2 \subseteq \mathbb{F}_q^{N_i}$ such that the following holds, where $g_i = \frac{N_i}{\sqrt{q} - 1} + \zeta$:
\begin{itemize}
\item $\cC_1$ has message length at least $a_1 - g_i + 1$ and distance at least $N_i - a_1$.
\item $\cC_2$ has message length at least $a_2 - g_i + 1$ and distance exactly $N_i - a_2$ and
\begin{align} \label{thm:min-codewords-bound}
A_{N_i - a_2}(\cC_2) \geq \frac{\binom{N_i}{a_2}}{(\sqrt{q} + 1)^{2g_i}}.
\end{align}
\end{itemize}  
Moreover, the generator matrices of $\cC_1, \cC_2$ can be computed in $O\left(\binom{N + a_2 - 1}{a_2} \cdot |\cC_2| \cdot \poly(N_i)\right)$ time.
\end{lemma}

\begin{proof}
Let $(N_i)_{i \in \mathbb{N}}$ be a dense sequence as in Theorem~\ref{thm:gs}. From Theorem~\ref{thm:vladut}, there exists an $\mathbb{F}_q$-positive divisor $D_2$ of degree $a_2$ such that the corresponding code $\cC_2 = C(X_i, D_2, S_i)$ (where $S \subseteq X_i(\mathbb{F}_q)$ of size $N_i$) satisfies~\eqref{thm:min-codewords-bound} and that its distance is $N_i - a_2$; from Theorem~\ref{thm:ag-basic}, its message length must also be at least $a_2-g_i+1$. Next, let $D_1$ be any $\mathbb{F}_q$-positive divisor of degree $a_1$ such that $D_2 - D_1 \geq 0$. Let $\cC_1 = C(X_i, D_1, S_i)$ be the corresponding AG code; once again, Theorem~\ref{thm:ag-basic} yields the desired bounds on its message length and distance. Finally, observe that $D_2 - D_1 \geq 0$ implies that $\cC_1 \subseteq \cC_2$ as desired.

The main bottleneck to algorithmically construct such codes lies in finding $D_2$. Nevertheless, the total number of degree-$a_2$ $\mathbb{F}_q$-positive divisor is only $\binom{N_i + a_2 - 1}{a_2}$. We can use brute force to enumerate all of them and check whether the corresponding code satisfies~\eqref{thm:min-codewords-bound}, which further takes $|\cC_2|$ time. This results in the claimed running time.
\end{proof}

Finally, we can now construct our gadgets, by an appropriate setting of parameters. In particular, $a_1$ and $a_2$ will be selected to be close to each other and to both be slightly larger than $N/\sqrt{q}$. This results in the graphs whose degrees are roughly square root of the number of vertices.

\begin{theorem} \label{thm:ag-gadget-new}
For every $0 < \delta < 1$, there exist $\mu > 0$ and a log-dense sequence $(n_i)_{i \in \mathbb{N}}$ such that, for every $i \in \mathbb{N}$, there is a bipartite graph $G_i = (A_i \dot \cup B_i,E_i)$ where $|A_i| = |B_i| = n_i$ and $|E_i| \geq \Omega(n_i^{2 - \delta})$, such that $(1 + \mu)$-$\cd(G) = O(\log n_i)$. Moreover, for all $i \in \mathbb{N}$, a $(\beta, 1 + \mu)$-gap-realization $\tau: A_i \dot \cup B_i \to \{0, 1\}^{O(\log n_i)}$ of $G_i$ can be constructed in time $O(n_i^3)$ for some $\beta = \Theta(\log n_i)$.
\end{theorem}

\begin{proof}
Once again, the proof here is similar to those of Theorems~\ref{thm:dense-cd} and~\ref{thm:mip-pre-tensor}, except that we use the (pairs of) AG codes from Lemma~\ref{lem:ag-pair} instead of Reed-Solomon codes. 


Let $q \geq 49$ be any sufficiently large square of prime and $\zeta > 0$ be any sufficiently small positive real number (both to be precisely specified later).

Let $(N_i)_{i \in \mathbb{N}}$ be the sequence guarantee by Lemma~\ref{lem:ag-pair}. Let $a_1 = N_i \cdot \left(\frac{1}{q^{0.5(1 - \delta)}} - \frac{1}{q}\right)$ and $a_2 = \frac{N_i}{q^{0.5(1 - \delta)}}$. For convenience, we assume that $a_1$ and $a_2$ are integers\footnote{Note that, for sufficiently large $N_i$, one can take the ceilings (or floors) of the specified values to get integers with negligible affect to the calculations.}. Let $\cC_1, \cC_2$ be the codes given by Lemma~\ref{lem:ag-pair}. The sequence $(n_i)_{i \in \mathbb{N}}$ is defined as $n_i = |\cC_1|$.


Applying Lemma~\ref{lem:finding-center} to $(\cC_1, \cC_2)$ implies that there exists $\bs \in \cC_2 \setminus \cC_1$ such that
\allowdisplaybreaks
\begin{align}
\frac{|\cB(\bs, \Delta(\cC_2)) \cap \cC_1|}{|\cC_1|}
&\geq \frac{A_{\Delta(\cC_2)}(\cC_2)}{|\cC_2|} \nonumber \\
(\text{From Lemma}~\ref{lem:ag-pair}) &\geq \frac{\binom{N_i}{a_2}}{(\sqrt{q} + 1)^{2g_i} \cdot |C_2|} \nonumber \\
(\text{Singleton Bound}) &\geq  \frac{\binom{N_i}{a_2}}{(\sqrt{q} + 1)^{2g_i} \cdot q^{a_2 + 1}} \nonumber \\
&\geq  \frac{\left(N_i/a_2\right)^{a_2}}{(\sqrt{q} + 1)^{2g_i} \cdot q^{a_2 + 1}} \nonumber \\
&= \frac{q^{0.5(1 - \delta)a_2}}{(\sqrt{q} + 1)^{2g_i} \cdot q^{a_2 + 1}} \nonumber \\
&= \frac{1}{(\sqrt{q} + 1)^{2g_i} \cdot q^{(0.5 + 0.5\delta)a_2 + 1}} \nonumber \\
&= \frac{1}{q^{(0.5 + 0.5\delta + o(1))a_2}} \nonumber \\
&= \frac{1}{q^{(0.5 + 0.5\delta + o(1))(a_1 + o(1))}} \nonumber \\
&= \frac{1}{|\cC_1|^{(0.5 + 0.5\delta + o(1))}} \nonumber \\
&\geq \Omega(|\cC_1|^{-0.5-0.5\delta - o(1)}) \label{eq:tmp1}
\end{align}
where $o(1)$ terms above denote the terms that go to zero as $q \to \infty$ and $\zeta \to 0$. As a result, by picking $q$ sufficiently large and $\zeta$ sufficiently small, the term in~\eqref{eq:tmp1} is at least $\Omega(|\cC_1|^{-0.5-\delta})$.

We construct the graph $G_i = (A_i, B_i, E_i)$ and a realization $\tau$ as follows. Let $A_i = \cC_1, B_i = \{\bs + \bc \mid \bc \in \cC_1\}$ and $E_i = \{(\ba, \bb) \in A_i \times B_i \mid \Delta(\ba, \bb) = \Delta(\cC_2)\}$. $G_i$ can be easily realized by applying the mapping $\psi: \mathbb{F}_q^q \to \{0, 1\}^{q^2}$ from Proposition~\ref{prop:simplex}. More precisely, let $\tau$ be the restriction of $\psi$ on $A_i \cup B_i$. Below we argue about the density of $G_i$ and that $\tau$ is a $(2\Delta(\cC_2), 1 + \mu)$-gap-realization of $G_i$ where $\mu = \frac{\Delta(\cC_1) - 1}{\Delta(\cC_2)} - 1$. Note that
\begin{align*}
\mu \geq \frac{a_2 - a_1 - 1}{N_i - a_2} = \Omega(1/q).
\end{align*}

Let us now check that $G_i$ and  $\tau$ satisfy all the claimed properties:
\begin{itemize}
\item First, notice that $|E_i|$ is exactly $|\cC_1| \cdot |\cB(\bs, \Delta(\cC_2)) \cap \cC_1| \geq \Omega(|\cC_1|^{1.5 - \delta}) = \Omega(n_i^{1.5 - \delta})$. 
\item For any $\bv_1 =\psi(\bc_1), \bv_2 = \psi(\bc_2)$ both from $X_i$ or both from $Y_i$, we have $\bc_1 - \bc_2 \in \cC_1 \setminus \{\bzero\}$. Hence, $\|\bv_1 - \bv_2\|_0 = 2 \cdot \Delta(\bv_1, \bv_2) \geq 2 \cdot \Delta(\cC_1) > (1 + \mu) \cdot (2 \Delta(\cC_2))$.
\item Next, for every $\ba \in A_i$ and $\bb \in B_i$, we have $\ba - \bb \in \cC_2 \setminus \{\bzero\}$. Thus, $\Delta(\ba, \bb) \geq \Delta(\cC_2)$. Hence, $\|\tau(\ba) - \tau(\bb)\|_0 = 2\Delta(\ba, \bb) \geq 2\Delta(\cC_2)$. Moreover, the inequality is an equality if and only if $\Delta(\ba, \bb) = \Delta(\cC_2)$, i.e., $(\ba, \bb) \in E_i$ as desired. 
\end{itemize}
Given $\cC_1, \cC_2$, the running time of constructing $(X_i, Y_i)$ is $O(|\cC_1| \cdot |\cC_2| \cdot q^2) = O(n_i^3)$. Moreover, the running time to construct $\cC_1$ and $\cC_2$, as given by Lemma~\ref{lem:ag-pair}, is
\begin{align*}
O\left(\binom{N + a_2 - 1}{a_2} \cdot |\cC_2| \cdot \poly(N_i)\right) &\leq O\left((e(N+a_2)/a_2)^{a_2} \cdot |\cC_2| \cdot \poly(N_i) \right) \\
&\leq O\left((2e\sqrt{q})^{a_2} \cdot |\cC_2| \cdot \poly(N_i) \right) \\
&\leq O\left(|\cC_1| \cdot |\cC_2| \cdot \poly(N_i) \right) \\
&\leq O(n_i^3),
\end{align*}
where the last two inequalities are true for any sufficiently large $q$.
\end{proof}

\section{Inapproximability of Maximum Inner Product}\label{sec:MIP}

In this section, we prove the hardness of approximating \MIP. Once again, we  show a stronger version (than Theorem~\ref{thm:MIP}) where every point has Boolean coordinates, as stated below.

\begin{theorem} \label{thm:MIP-01}
Assuming \OVH, for every $\varepsilon>0$, there is no algorithm running in $O(n^{2-\varepsilon})$ time for $\gamma$-\MIP even for points in $\{0, 1\}^{n^{o(1)}}$, for any $\gamma\le 2^{(\log n)^{1-o(1)}}$.
\end{theorem}

The proof proceeds in two steps: first, we show hardness of approximating \MIP in low dimension but with a small ($1 + o(1)$) approximation factor. Second, we  use tensor product operation to amplify the gap to be almost polynomial, as stated in Theorem~\ref{thm:MIP-01}. More specifically, in the first step, we prove the following:
\begin{theorem} \label{thm:MIP-pre-tensor}
Assuming \OVH, for every $\varepsilon > 0$, there exists $s_{\varepsilon} > 0$ such that no algorithm running in $O(n^{2-\varepsilon})$  time can solve $\left(1 + \frac{1}{\log \log n}\right)$-\MIP even for points in $\{0, 1\}^{\left(\log n\right)^{s_{\varepsilon}}}$.
\end{theorem}

Note that the factor $\frac{1}{\log \log n}$ is not significant, and this can be replaced by any $o(1)$ factor; we use this just to make the calculations more concrete. Before we move on to the proof of Theorem~\ref{thm:MIP-pre-tensor}, let us first show how it implies Theorem~\ref{thm:MIP-01}.

\begin{proof}[Proof of Theorem~\ref{thm:MIP-01} from Theorem~\ref{thm:MIP-pre-tensor}]
Let $(P, \alpha)$ be an instance of $\left(1 + \frac{1}{\log \log n}\right)$-\MIP where $P \subseteq \{0, 1\}^{\left(\log n\right)^{s_{\varepsilon}}}$. For $t = \frac{\log n}{(\log \log n)^2}$, define $P' = \{\bx^{\otimes t} \mid \bx \in P\}, \alpha' = \alpha^t$ and $\gamma = \left(1 + \frac{1}{\log \log n}\right)^t = 2^{(\log n)^{1 - o(1)}}$. The dimension of points in $P'$ is $(\log n)^{s_{\varepsilon} \cdot t} = n^{o(1)}$. Moreover, it is easy to check, based on the identity $\left<\bx^{\otimes t}, \by^{\otimes t}\right> = \left<\bx, \by\right>^t$, that $(P', \alpha')$ is a YES (resp. no) instance of $\gamma$-MIP iff $(P, \alpha)$ is a YES (resp. NO) instance of $\left(1 + \frac{1}{\log \log n}\right)$-\MIP.

In other words, if there is an $O(n^{2 - \varepsilon})$ time algorithm for $\gamma$-\MIP in $n^{o(1)}$ dimension, then there also exist an $O(n^{2 - \varepsilon})$ subquadratic time algorithm for $\left(1 + \frac{1}{\log \log n}\right)$-\MIP in $(\log n)^{s_\varepsilon}$ dimension. Thus, Theorem~\ref{thm:MIP-01} follows from Theorem~\ref{thm:MIP-pre-tensor}. 
\end{proof}

The rest of this section is devoted to proving Theorem~\ref{thm:MIP-pre-tensor}. To do so, we consider the  gap-\ABMIP problem.

\begin{definition}[$\gamma$-\ABMIP problem]
 Let $\gamma\ge 0$. In the $\gamma$-\ABMIP problem we are given two sets $A,B$ each of $n$ points in $\{0, 1\}^d$ and an integer $\alpha\in[d]$ as input, and the goal is to distinguish between the following two cases.
\begin{itemize}
\item \textbf{Completeness.} There exists $(a,b)\in A\times B$ such that $\langle a,b\rangle\ge \alpha$.
\item \textbf{Soundness.} For every $(a,b)\in A\times B$ we have  $\langle a,b\rangle< \alpha - \gamma$.
\end{itemize}
\end{definition}

We need the below hardness result from~\cite{R18}. Note that the result is stated differently in~\cite{R18}; for how the result in~\cite{R18} implies the one below, see Section 3.2 of~\cite{C18}.
\begin{theorem}[\cite{R18}] \label{thm:add-r18}
Assuming \OVH, for every $\varepsilon > 0$, there is no algorithm running in $O(n^{2-\varepsilon})$ time for the $\gamma$-\ABMIP problem, for any $d=\omega(\log n)$ and $\gamma=o(d)$.
\end{theorem}

\begin{proof}[Proof of Theorem~\ref{thm:MIP-pre-tensor}]
For any $\varepsilon > 0$, let $C_{\text{exp}}$ be the constant such that the dimension of $\tau$ in Theorem~\ref{thm:mip-pre-tensor} is at most $(\log n_i)^{C_{\text{exp}}/\varepsilon}$ for $\delta = \varepsilon/2$. We define $s_\varepsilon$ as $2 \cdot C_{\text{exp}} / \varepsilon + 2$.

Suppose contrapositively that there exists $\varepsilon > 0$ and an algorithm $\cA$ that can solve $\left(1 + \frac{1}{\log \log n}\right)$-\MIP of dimension $(\log n)^{s_{\varepsilon}}$ in time $n^{2 - \varepsilon}$. We will construct an algorithm $\cA'$ that solves $(\log n)$-\ABMIP in time $n^{2 - \varepsilon'}$ for some constant $\varepsilon' > 0$ (to be specified below) for $d = (\log n \sqrt{\log \log n})$ dimensions. Together with Theorem~\ref{thm:add-r18}, this implies that \OVH is false, as desired.

Let $C_\varepsilon$ denote the constant of the log-dense sequence from Theorem~\ref{thm:mip-pre-tensor} for $\delta = \varepsilon/2$, and let $\varepsilon'$ be $0.01 \cdot \varepsilon / C_{\varepsilon}$. The algorithm $\cA'$ on input $(A, B, \alpha)$ where $A, B \subseteq \{0, 1\}^d, \alpha \in [d]$ works as follows:
\begin{enumerate}
\item Let $n'$ be the largest number in the sequence from Theorem~\ref{thm:mip-pre-tensor} with $\delta = \varepsilon/2$ s.t. $n' \leq n^{0.1}$.
\item Let $G' = (A' \dot\cup B', E')$ be the graph from Theorem~\ref{thm:mip-pre-tensor} with $|A'| = |B'| = n'$, $|E'| \geq \Omega((n')^{2-\delta})$, and $\tau: A'\dot\cup B' \to \{0, 1\}^{(\log n')^{C_{\text{exp}}/\varepsilon}}$ be a $(\beta, 3)$-gap-\IP-relization of $G'$ where $\beta \in \mathbb{N}$.
\item We use the algorithm from Lemma~\ref{lem:cover} to find $\pi_1, \dots, \pi_k$ where $k = O((n')^{\delta}\log n')$ such that the union of $E_{G'_{\pi_1}}, \ldots, E_{G'_{\pi_k}}$ is $E_{K_{n', n'}}$
\item We assume w.l.o.g. that $n$ is divisible by $n'$. Partition $A$ and $B$ into $A_1, \dots, A_{n/n'}$ and $B_1, \dots, B_{n/n'}$ each of size $n'$. For each $i, j \in [n/n'], t \in [k]$, do the following:
\begin{enumerate}
\item 
Let $\tau_t$ be an appropriate permutation of $\tau$ that $(\beta, 3)$-gap-\IP-realizes $G'_{\pi_t}$.
\item Let $\alpha' = \beta \cdot \alpha + 3d \cdot \beta$, and define $A_i^t, B_j^t$ as
$$
A_i^t = \{(\bone_{\beta} \otimes \ba) \circ (\bone_{3d} \otimes \tau_t(\ba)) \mid \ba \in A_i\}, B_j^t = \{(\bone_{\beta} \otimes \bb) \circ (\bone_{3d} \otimes \tau_t(\bb)) \mid \bb \in B_j\}.
$$
\item Run $\cA$ on $(A_i^t \dot\cup B_j^t, \alpha')$. If $\cA$ outputs YES, then output YES and terminate.
\end{enumerate}
\item If none of the executions of $\cA$ returns with YES, then output NO.
\end{enumerate}

Observe that the bottleneck in the running time of the algorithm is in the executions of $\cA$. The number of executions is $(n/n')^2 \cdot k$ and each execution takes $O((n')^{2 - \varepsilon})$ time. Hence, in total the running time of the algorithm $\cA'$ is $O((n/n')^2 \cdot k \cdot (n')^{2 - \varepsilon}) \leq O(n^2 \log n \cdot (n')^{-\varepsilon/2})$. Now, from the log-density of the sequence from Theorem~\ref{thm:mip-pre-tensor}, we have $n' \geq n^{0.1/C_\varepsilon} = n^{10\varepsilon'/\varepsilon}$. As a result, the running time of $\cA$ is at most $O(n^{2 - 5\varepsilon'}\log n) \leq O(n^{2 - \varepsilon'})$ as desired.

To see the correctness of the algorithm, first observe that the dimensions of vectors in $A_i^t, B_j^t$ are at most $\beta \cdot d + 3 d \cdot (\log n')^{C_{\text{exp}}/\varepsilon}$ which is at most $(\log {n})^{s_{\varepsilon}}$ for any sufficiently large $n$; that is, the calls to $\cA$ are valid. Next, observe that, if $(A, B, \alpha)$ is a YES instance of \ABMIP, there must be $i, j \in [n/n']$ and $\ba^* \in A_i, \bb^* \in B_j$ such that $\left<\ba^*, \bb^*\right>$ is at least $\alpha$. Since $G'_{\pi_1}, \dots, G'_{\pi_k}$ covers $K_{n', n'}$, there must be $t \in [k]$ such that $\left<\tau_t(\ba^*), \tau_t(\bb^*)\right> \geq \beta$. As a result, $\left<(\bone_\beta \otimes \ba^*) \circ (\bone_{3d} \otimes \tau_t(\ba^*), (\bone_\beta \otimes \bb^*) \circ (\bone_{3d} \otimes \tau_t(\bb^*))\right> \geq \beta \cdot \alpha + 3 d \cdot \beta = \alpha'$. Thus, $(A_i^t \cup B_j^t, \alpha')$ is a YES instance for \MIP and $\cA'$ outputs YES as desired.

Finally, let us assume that $(A, B, \alpha)$ is a NO instance of $(\log n)$-\ABMIP. Consider any $i, j \in [n/n']$ and $t \in [k]$. To argue that $(A_i^t \cup B_j^t, \alpha')$ is a NO instance for $\left(1 + \frac{1}{\log \log {n'}}\right)$-\MIP, we have to show that any two points in $A_i^t \cup B_j^t$ have inner product less than $\alpha'/\left(1 + \frac{1}{\log \log {n'}}\right)$. To see this, let us consider two cases.
\begin{enumerate}
\item The two points are either both from $A_i^t$ or both from $B_j^t$. Assume w.l.o.g. that the two points are from $A_i^t$; let them be $(\bone_\beta \otimes \ba) \circ (\bone_{3d} \otimes \tau_t(\ba))$ and $(\bone_\beta \otimes \ba') \circ (\bone_{3d} \otimes \tau_t(\ba'))$. Recall that, from Theorem~\ref{thm:mip-pre-tensor}, we must have $\left<\tau_t(\ba), \tau_t(\ba')\right> < \beta/3$. Moreover, since $\ba, \ba' \in \{0, 1\}^d$, we have $\left<\ba, \ba'\right> \leq d$. Thus, we can conclude that
\begin{align*}
\left<(\bone_\beta \otimes \ba) \circ (\bone_{3d} \otimes \tau_t(\ba)), (\bone_\beta \otimes \ba') \circ (\bone_{3d} \otimes \tau_t(\ba'))\right> &< \beta \cdot d + 3d \cdot (\beta/3) \\
&< (2/3) \cdot \alpha',
\end{align*}
which is less than $\alpha'/\left(1 + \frac{1}{\log \log {n'}}\right)$ for any sufficiently large $n$.
\item One of the point is from $A_i^t$ and the other from $B_j^t$. Let them be $(\bone_\beta \otimes \ba) \circ (\bone_{3d} \otimes \tau_t(\ba))$ and $(\bone_\beta \otimes \bb) \circ (\bone_{3d} \otimes \tau_t(\bb))$. Since $(A, B, \alpha)$ is a NO instance of $(\log n)$-\ABMIP, we must have $\left<\ba, \bb\right> < \alpha - \log n$. Furthermore, from Theorem~\ref{thm:mip-pre-tensor}, we must have $\left<\tau_t(\ba), \tau_t(\bb)\right> \leq \beta$. Combining the two implies that 
\begin{align*}
\left<(\bone_\beta \otimes \ba) \circ (\bone_{3d} \otimes \tau_t(\ba)), (\bone_\beta \otimes \bb) \circ (\bone_{3d} \otimes \tau_t(\bv))\right> &< \beta \cdot (\alpha - \log n) + 3d \cdot \beta \\
&= \alpha' - \beta \cdot (\log n) \\
(\text{Since } \alpha' \leq 4d\beta) &\leq \alpha' \left(1 - \frac{1}{4\sqrt{\log \log n}}\right) \\
&\leq \alpha'\left(1 - \frac{1}{\log \log n'}\right) \\
&\leq \alpha'/\left(1 + \frac{1}{\log \log n'}\right),
\end{align*} 
where the second-to-last inequality holds for any sufficiently large $n$.
\end{enumerate}

Hence, $(A_i^t \dot\cup B_j^t, \alpha')$ must be a NO instance for $\left(1 + \frac{1}{\log \log {n'}}\right)$-\MIP for every $t\in[k]$ and $i,j\in [n/n']$. Thus, $\cA'$ outputs NO as desired.
\end{proof}


\section{Inapproximability of Closest Pair}\label{sec:ACP}

In this section, we prove the hardness of approximating \CP (Theorem~\ref{thm:gapCP}). As usual, we reduce from the bichromatic version of the problem, and the lower bound for the bichromatic version is stated below:

\begin{theorem}[Rubinstein \cite{R18}]\label{thm:R18}
Assuming \OVH, for every $\varepsilon>0$ there exists $\kappa > 0$ such that there is no algorithm running in $n^{2-\varepsilon}$ time for $(1 + \kappa)$-\BCP in the Hamming metric. Moreover, this holds even for instances $(A, B, \alpha)$ of $(1 + \kappa)$-\BCP when $d = \Theta_\varepsilon(\log n), \alpha = \Theta_\varepsilon(\log n)$ and $A, B \subseteq \{0, 1\}^d$.
\end{theorem}

Again, we prove below the inapproximability of the gap-\CP problem  for Boolean vectors. Clearly, this immediately implies Theorem~\ref{thm:gapCP}.

\begin{theorem}\label{thm:gapCP01}
Assuming \OVH, for every $\varepsilon>0$, there exists $\theta > 0$ and $c > 0$ such that there is no algorithm running in $n^{1.5-\varepsilon}$ time for $(1 + \theta)$-\CP in the Hamming metric for point-set in $\{0, 1\}^{c \cdot \log n}$.
\end{theorem}

\begin{proof}
Assume towards a contradiction that there exists an $\varepsilon > 0$ and an algorithm $\cA$ that, for every $\theta > 0$ solves $(1 + \theta)$-\CP of dimension $c \cdot \log n$ in time $O(n^{1.5 - \varepsilon})$, where $c:=c(\varepsilon)$ is a constant that will be specified later. Let $\varepsilon' > 0$ be a small constant (depending on $\varepsilon$) that we will specify below and let $\kappa = \kappa(\varepsilon')$ be as in Theorem~\ref{thm:R18}. We construct below an algorithm $\cA'$ that solves $(1 + \kappa)$-\BCP in time $O(n^{2 - \varepsilon'})$ for any instance $(A, B, \alpha)$ such that $A, B \subseteq \{0, 1\}^{O(\log n)}$ and $\alpha = \Theta(\log n)$. Together with Theorem~\ref{thm:R18}, this implies that \OVH is false, as desired.

Let $C_\varepsilon$ denote the constant of the log-dense sequence from Theorem~\ref{thm:ag-gadget-new} for $\delta = \varepsilon/2$, and let $\varepsilon'$ be $0.01 \cdot \varepsilon / C_{\varepsilon}$. Let $\mu$ be the constant from Theorem~\ref{thm:ag-gadget-new}. Select $\theta > 0$ be a sufficiently small constant such that $\frac{\mu - \theta}{1 + \theta} > \frac{\theta}{\kappa - \theta}$.

The algorithm $\cA'$ on $(A, B, \alpha)$ where $A, B \subseteq \{0, 1\}^{O(\log n)}, \alpha = \Theta(\log n)$ works as follows:
\begin{enumerate}
\item Let $n'$ be the largest number in the sequence from Theorem~\ref{thm:ag-gadget-new} with $\delta = \varepsilon/2$ s.t. $n' \leq n^{0.1}$.
\item Let $G' = (A' \dot\cup B', E')$ be the graph from Theorem~\ref{thm:ag-gadget-new} with $|A'| = |B'| = n'$, $|E'| \geq \Omega((n')^{1.5-\delta})$, and $\tau: A'\dot\cup B' \to \{0, 1\}^{O(\log n')}$ be a $(\beta, 1 + \mu)$-gap-relization of $G'$ where $\beta \in \mathbb{N}$ and $\beta = \Theta(\log n')$.
\item We use the algorithm from Lemma~\ref{lem:cover} to find $\pi_1, \dots, \pi_k$ where $k = O((n')^{0.5 + \delta}\log n')$ such that the union of $E_{G'_{\pi_1}}, \ldots, E_{G'_{\pi_k}}$ is $E_{K_{n', n'}}$
\item We assume w.l.o.g. that $n$ is divisible by $n'$. Partition $A$ and $B$ into $A_1, \dots, A_{n/n'}$ and $B_1, \dots, B_{n/n'}$ each of size $n'$. For each $i, j \in [n/n'], t \in [k]$, do the following:
\begin{enumerate}
\item Let $\tau_t$ be an appropriate permutation of $\tau$ that $(\beta, 1 + \mu)$-gap-realizes $G'_{\pi_t}$.
\item Pick $r_1, r_2$ such that 
\begin{align} \label{eq:sel}
\frac{\theta}{\kappa - \theta} \cdot \frac{\beta}{\alpha} \leq \frac{r_1}{r_2} \leq \frac{\mu - \theta}{1 + \theta} \cdot \frac{\beta}{\alpha}.
\end{align}
Notice that the upper and lower bounds are $\Theta(1)$ and they are also $\Theta(1)$ apart. Hence, we can pick these $r_1, r_2$ so that $r_1, r_2 = \Theta(1)$. 
\item Let $\alpha' = r_1 \cdot \alpha  + r_2 \cdot \beta$ and define $A_i^t, B_j^t$ as
\begin{align*}
A_i^t = \{(\bone_{r_1} \otimes \ba) \circ (\bone_{r_2} \otimes \tau_t(\ba)) \mid \ba \in A_i\}, B_j^t = \{(\bone_{r_1} \otimes \bb) \circ (\bone_{r_2} \otimes \tau_t(\bb)) \mid \bb \in B_j\}.
\end{align*}
\item Run $\cA$ on $(A_i^t \cup B_j^t, \alpha')$. If $\cA$ outputs YES, then output YES and terminate.
\end{enumerate}
\item If none of the executions of $\cA$ returns with YES, then output NO.
\end{enumerate}

Observe that the bottleneck in the running time of the algorithm is in the executions of $\cA$. The number of executions is $(n/n')^2 \cdot k$ and each execution takes $O((n')^{1.5 - \varepsilon})$ time. Hence, in total the running time of the algorithm $\cA'$ is $O((n/n')^2 \cdot k \cdot (n')^{1.5 - \varepsilon}) \leq O(n^2 \log n \cdot (n')^{-\varepsilon/2})$. Now, from the log-density of the sequence from Theorem~\ref{thm:ag-gadget-new}, we have $n' \geq n^{0.1/C_\varepsilon} = n^{10\varepsilon'/\varepsilon}$. As a result, the running time of $\cA$ is at most $O(n^{2 - 5\varepsilon'}\log n) \leq O(n^{2 - \varepsilon})$ as desired.

To see the correctness of the algorithm, first observe that the dimensions of vectors in $A_i^t, B_j^t$ are at most $r_1 \cdot \alpha + r_2 \cdot \beta$ which is $O(\log {n'})$; that is, the calls to $\cA$ are valid. Next, observe that, if $(A, B, \alpha)$ is a YES instance of \BCP, there must be $i, j \in [n/n']$ and $\ba^* \in A_i, \bb^* \in B_j$ such that $\|\ba^* - \bb^*\|_0$ is at most $\alpha$. Since $G'_{\pi_1}, \dots, G'_{\pi_k}$ covers $K_{n', n'}$, there must be $t \in [k]$ such that $\|\tau_t(\ba^*) - \tau_t(\bb^*)\|_0 \leq \beta$. As a result, $\|((\bone_{r_1} \otimes \ba^*) \circ (\bone_{r_2} \otimes \tau_t(\ba^*)) - ((\bone_{r_1} \otimes \bb^*) \circ (\bone_{r_2} \otimes \tau_t(\bb^*)))\|_0 \leq r_1 \cdot \alpha + r_2 \cdot \beta = \alpha'$. Thus, $(A_i^t \cup B_j^t, \alpha')$ is a YES instance for \CP and $\cA'$ outputs YES as desired.

Finally, let us assume that $(A, B, \alpha)$ is a NO instance of $(1 + \kappa)$-\BCP. Consider any $i, j \in [n/n']$ and $t \in [k]$. To argue that $(A_i^t \cup B_j^t, \alpha')$ is a NO instance for $(1 + \theta)$-\CP, we have to show that any two points in $A_i^t \cup B_j^t$ have distance more than $\alpha'$. To see this, let us consider two cases.
\begin{enumerate}
\item Both points are either from $A_i^t$ or from $B_j^t$. Assume w.l.o.g. that they are from $A_i^t$; let them be $(\bone_{r_1} \otimes \ba) \circ (\bone_{r_2} \otimes \tau_t(\ba))$ and $(\bone_{r_1} \otimes \ba') \circ (\bone_{r_2} \otimes \tau_t(\ba'))$. Recall that, from the definition of $X'_t$ and Theorem~\ref{thm:ag-gadget-new}, we must have $\|\tau_t(\ba) - \tau_t(\ba')\|_0 > (1 + \mu) \cdot \beta$. Thus, the Hamming distance between the two points is more than $r_2 \cdot (1 + \mu) \cdot \beta \geq (1 + \theta) \cdot \alpha'$, where the inequality comes from our choice of $r_1, r_2$.
\item One of the point is from $A_i^t$ and the other from $B_j^t$. Let them be $(\bone_{r_1} \otimes \ba) \circ (\bone_{r_2} \otimes \tau_t(\ba))$ and $(\bone_{r_1} \otimes \bb) \circ (\bone_{r_2} \otimes \tau_t(\bb))$. Since $(A, B, \alpha)$ is a NO instance of $(1 + \kappa)$-\BCP, $\|\ba - \bb\|_0 > (1 + \kappa) \cdot \alpha$. Moreover, from definition of $\tau_t$, we must have $\|\tau_t(\ba) - \tau_t(\bb)\|_0 \geq \beta$. Combining the two implies that the distance between $(\bone_{r_1} \otimes \ba) \circ (\bone_{r_2} \otimes \tau_t(\ba))$ and $(\bone_{r_1} \otimes \bb) \circ (\bone_{r_2} \otimes \tau_t(\bb))$ is more than $r_1 \cdot (1 + \kappa) \cdot \alpha + r_2 \cdot \beta \geq (1 + \theta) \cdot \alpha'$, where the inequality is once again from our choice of $r_1, r_2$. 
\end{enumerate}
Hence, $(A_i^t \dot\cup B_j^t, \alpha')$ must be a NO instance for $(1+\theta)$-\CP for every $t\in[k]$ and $i,j\in [n/n']$. Thus, $\cA'$ outputs NO as desired.
\end{proof}

\section{Discussion and Open Questions}\label{sec:open}

It remains open to completely resolve Open Questions~\ref{open:Q1} and~\ref{open:Q2}. It is still possible that our framework can be used to resolve these problems: we just need to construct gadgets with better parameters! In particular, to resolve Question~\ref{open:Q1}, we have to improve the dimension bound in Theorem~\ref{thm:dense-cd} to $O_{\delta}(\log n_i)$. For Question~\ref{open:Q2}, we just have to improve the bound on the number of pairs in (3) of Theorem~\ref{thm:ag-gadget-new} to $\Omega(n_i^{2 - \delta})$. Following our observation from Lemma~\ref{lem:finding-center}, this motivates us to ask the following purely coding theoretic question:

\begin{open}
For every $0<\delta<1$, are there linear codes $\cC_1\subseteq \cC_2\subseteq \mathbb F_q^N$ both of block length $N$ over alphabet $\mathbb{F}_q$ such that the following holds:
\begin{itemize}
\item $\Delta(\cC_1)\ge (1+f(\delta))\cdot \Delta(\cC_2)$, for some $f:(0,1)\to (0,1)$.
\item $|A_{\Delta(\cC_2)}(\cC_2)|/|\cC_2|\ge |\cC_1|^{-\delta}$.
\end{itemize}  
\end{open}

Apart from the aforementioned questions, Rubinstein \cite{R18} pointed out an interesting obstacle, aptly dubbed the ``triangle inequality barrier'', to obtain fine-grained lower bounds against 3-approximation algorithms for \BCP (see Open Question 3 in \cite{R18}). In the case of \CP, this barrier turns out to be against 2-approximation algorithms as noted in \cite{DKL18}. We reiterate this below as an open problem to be resolved:

\begin{open}
Can we show that assuming \SETH, for some constant $\varepsilon > 0$, no algorithm running in time $n^{1+\varepsilon}$ can solve 2-\CP in \emph{any} metric when the points are in $\omega(\log n)$ dimensions?
\end{open}

Another interesting direction is to extend the hardness of \MIP to the $k$-vector generalization of the problem, called $k$-\MIP. In $k$-\MIP, we are given a set of $n$ points $P \subseteq \mathbb{R}^d$ and we would like to select $k$ distinct points $\ba_1, \dots, \ba_k \in P$ that maximizes
\begin{align*}
\left<\ba_1, \dots, \ba_k\right> := \sum_{j \in [d]} (\ba_1)_j \cdots (\ba_k)_j.
\end{align*}

It is known that the $k$-chromatic variant of $k$-\MIP is hard to approximate (see Appendix B of~\cite{KLM18}) but this is not known to be true for $k$-\MIP itself. Our approach seems quite compatible to tackling this problem as well; in particular, if we can construct a certain (natural) generalization of our gadget for \MIP, then we would immediately arrive at the inapproximability of $k$-\MIP even for $\{0, 1\}$-entries vectors. The issue in constructing this gadget is that we are now concerned about agreements of more than two vectors, which does not correspond to error-correcting codes anymore and some additional tools are needed to argue  for this more general case.

It should be noted that the hardness of approximating $k$-\MIP for $\{0, 1\}$-entry vectors is equivalent to the \emph{one-sided $k$-biclique} problem~\cite{Lin15}, in which a bipartite graph is given and the goal is to select $k$ vertices on the right that  maximize the number of their common neighbors. The equivalence can be easily seen by viewing the coordinates as the left-hand-side vertices and the vectors as the right-hand-side vertices. The one-sided $k$-biclique is shown to be \Wone{}-hard to approximate by Lin~\cite{Lin15} who also showed a lower bound of $n^{\Omega(\sqrt{k})}$ for the problem assuming \ETH. If the generalization of our gadget for $k$-\MIP works as intended, then this lower bound can be improved to $n^{\Omega(k)}$ under \ETH and even $n^{k - o(1)}$ under \SETH{}.

The one-sided $k$-biclique is closely related to the (two-sided) $k$-biclique problem, where we are given a bipartite graph and we wish to decide whether it contains $K_{k, k}$ as a subgraph. The $k$-biclique problem was consider a major open problem in parameterized complexity (see e.g.,~\cite{DowneyF13}) until it was shown by Lin to be \Wone{}-hard~\cite{Lin15}. Nevertheless, the running time lower bound known is still not tight: currently, the best lower bound known for this problem is $n^{\Omega(\sqrt{k})}$ both for the exact version (under \ETH{})~\cite{Lin15} and its approximate variant (under \GapETH{})~\cite{CCKLMNT17}. It remains an interesting open question to close the gap between the above lower bounds and the trivial upper bound of $n^{O(k)}$. Progresses on the one-sided $k$-biclique problem could lead to improved lower bounds for $k$-biclique problem too, although several additional steps have to be taken care of.

\subsection*{Acknowledgements}
We are grateful to Madhu Sudan for extremely helpful and informative discussion about AG codes; in particular, Madhu pointed us to~\cite{vluaduct2018lattices}. We thank Bundit Laekhanukit and Or Meir for general discussions, and the Simons Institute for their wonderful work-space. Finally, we would like to thank Lijie Chen for sharing \cite{CW19}, and Orr Paradise for useful comments on an earlier draft of this manuscript. 

\addcontentsline{toc}{section}{\protect\numberline{}References}%

\bibliographystyle{alpha}
\bibliography{ref}

\appendix

\section{Lower Bound on Gap Closest Pair in Edit Distance Metric}\label{sec:edit}

In this section we prove Theorem~\ref{thm:CP-Edit}. The proof is almost identical to Rubinstein's \cite{R18} proof for the \OVH-hardness of gap-\BCP in the edit distance metric and uses the following technical tool established in \cite{R18}.

\begin{lemma}[Rubinstein \cite{R18}]\label{lem:edit}
For large enough $d\in\mathbb{N}$, there is a function $\zeta:\{0,1\}^d\to \{0,1\}^{d'}$, where $d'=O(d\log d)$, such that for all $a,b\in\{0,1\}^d$ the following holds for some constant $\lambda>0$:
$$
\left|\ed(\zeta(a),\zeta(b))-\lambda\cdot \log d\cdot \|a-b\|_0\right|= o(d').
$$
Moreover, for any $a\in \{0,1\}^d$, $\zeta(a)$ can be computed in $2^{o(d)}$ time.
\end{lemma}

At a high level, $\zeta$ picks a random $O(\log d)$-bit string $s_{i,x}$ uniformly and independently for every $(i,x)\in[d]\times \{0,1\}$, and for every vector $u\in\{0,1\}^d$, replaces the $i^{\text{th}}$ coordinate $u_i$ by $s_{i,u_i}$. The claims in the lemma statement follow by the known concentration bounds on the edit distance of random strings \cite{M89,L09}. This construction is further efficiently derandomized  by using $\log d$-wise independent strings  \cite{K13}.

\begin{proof}[Proof of Theorem~\ref{thm:CP-Edit}]
We show that if there exists an algorithm $\cA$ running in time $O(n^{1.5 - \varepsilon})$ for some $\varepsilon>0$ that can solve $(1+\delta)$-\CP in the edit distance metric for some $\delta>0$ over point-sets in $\{0,1\}^{d'}$, then $\cA$ can be used to solve $(1+\delta-o(1))$-\CP in the Hamming metric in time $O(n^{1.5 - \varepsilon})$ over point-sets in $\{0,1\}^d$, where $d'=O(d\log d)$.   Together with Theorem~\ref{thm:gapCP01}, this implies that \OVH is false, as desired.

Let $(P,\alpha)$ be an instance of $(1+\delta)$-\CP in the Hamming metric over point-sets in $\{0,1\}^{d}$. It is clear\footnote{In fact, one can design a $2^{\alpha}\cdot n\log n$ time algorithm for \CP in the Hamming metric, and therefore to assume \OVH, we require $\alpha=\Omega(d)$.} from the proofs of Theorem~\ref{thm:R18} and Theorem~\ref{thm:gapCP01} that $\alpha=\Omega(d)$. We now define an instance of $(P',\alpha':=(1+o(1))\cdot \lambda\log d\cdot \alpha)$ of $(1+\delta-o(1))$-\CP in the edit distance metric as follows.
Recall the function $\zeta$ from Lemma~\ref{lem:edit} and define the set $P'=\{\zeta(p)\mid p\in P\}$. Notice that for every pair of distinct points $p,q\in P$, we have $\left|\ed(\zeta(p),\zeta(q))=\lambda\cdot \log d\cdot \|p-q\|_0 \right|= o(d')$. In other words if we had a pair of distinct points $p,q$ in $P$ such that $\|p-q\|_0\le \alpha$ then, $\ed(\zeta(p),\zeta(q))\le\lambda \log d\cdot  \alpha + o(d')=(1+o(1))\cdot \lambda \log d\cdot  \alpha$ and suppose for all pairs of distinct points $p,q\in P$ we had $\|p-q\|_0> (1+\delta)\cdot \alpha$ then $\ed(\zeta(p),\zeta(q))>\lambda \log d\cdot  (1+\delta)\cdot \alpha - o(d')>(1+\delta -o(1))\lambda \log d\cdot  \alpha$, since $\alpha=\Omega(d)$.
This completes the analysis of the completeness and soundness cases, and we can conclude that running $\cA$ on input $(P',\alpha')$ solves the instance $(P,\alpha)$ of $(1+\delta)$-\CP in the Hamming metric.
\end{proof}

\section{Covering Biclique By Isomorphic Graphs: Proof of Lemma~\ref{lem:cover}} \label{app:biclique-cover}

Below we prove Lemma~\ref{lem:cover}. The proof strategy is similar to how the greedy approximation algorithms for the set cover problem are analyzed: we show that at each step, we can pick a graph isomorphic to $G$ that covers at least $|E_G|/n^2$ fraction of the remaining edges of the biclique. By doing so, we guarantee that the process ends in $O(\log n) \cdot n^2/|E_G|$ steps. Note however that, there are exponential number of isomorphisms and thus we cannot simply enumerate all isomorphisms to find one that covers the desired fraction of uncovered edges. Nevertheless, it is not hard to see that we can use the method of conditional expectation to find one such isomorphism in polynomial time. This is formalized below.

\begin{lemma} \label{lem:derandomized}
For any two bipartite graphs $G = (A \dot\cup B, E_G)$ and $H = (A \dot\cup B, E_H)$, there exists a side-preserving permutation $\pi: A\dot\cup B \to A\dot\cup B$ such that $$|E_H \cap E_{G_{\pi}}| \geq \frac{|E_G| \cdot |E_H|}{|A| \cdot |B|}.$$
Moreover, such a permutation  $\pi$ can be found (deterministically) in $O((|A| + |B|)^4)$ time.
\end{lemma}

\begin{proof}
Notice that, if we pick $\pi|_A$ and $\pi|_B$ randomly among all permutations of $A$ and $B$ respectively, then, for a fixed $(a, b) \in E_H$, the probability that $(a, b)$ belongs to $E_{G_{\pi}}$ is $\frac{|E_G|}{|A| \cdot |B|}$. Thus,
\begin{align*}
\E_{\pi}\left[|E_H \cap E_{G_{\pi}}|\right] = \frac{|E_G| \cdot |E_H|}{|A| \cdot |B|}.
\end{align*}
This proves the existence part of the claim. To deterministically find such a $\pi$, we use the method of conditional expectation. Suppose $A \dot\cup B = \{1, \dots, n\}$. The algorithm works as follows:
\begin{enumerate}
\item Let $V_{\text{assigned}} \leftarrow \emptyset$.
\item For $i = 1, \dots, n$:
\begin{enumerate}
\item If $i \in A$, let $V_{\text{candidate}} = A \setminus V_{\text{assigned}}$. Otherwise, if $i \in B$, let $V_{\text{candidate}} = B \setminus V_{\text{assigned}}$.
\item For each $k \in V_{\text{candidate}}$, compute the conditional expectation: $$\E_{\pi}\left[|E_H \cap E_{G_\pi}| \,\middle\vert\, \pi(i) = k \wedge \left(\bigwedge_{j=1}^{i-1} \pi(j) = \pi^*(j)\right)\right].$$ Let $k^*$ be the maximizer for the above conditional expectation. We set $\pi^*(i) = k^*$.
\end{enumerate}
\item Output $\pi^*$.
\end{enumerate}
It is simple to see that the conditional expectation never decreases as we fill in the permutation. As a result, we must have $|E_H \cap E_{G_\pi}| \geq \frac{|E_G| \cdot |E_H|}{|A| \cdot |B|}$ as desired. Moreover, it is easy to see that the conditional expectation can be computed in time $O(|A| \cdot |B|)$ because, for each edge $(a, b) \in E_H$, we can compute the probability that $(a, b) \in E_{G_{\pi}}$ in $O(1)$ time. As a result, the overall running time of the algorithm is $O((|A| + |B|)^4)$.
\end{proof}

Finally using Lemma~\ref{lem:derandomized}, we prove Lemma~\ref{lem:cover} using the strategy outlined earlier in this section.

\begin{proof}[Proof of Lemma~\ref{lem:cover}]
We describe below an algorithm for finding $\pi_1, \dots, \pi_k$. It works as follows.
\begin{enumerate}
\item Let $k \leftarrow 0$.
\item While $E_H := E_{K_{n,n}} \setminus \underset{i\in [k]}{\cup}E_{G_{\pi_i}}$ is non-empty, do the following:
\begin{enumerate}
\item Let $k \leftarrow k + 1$.
\item Let $H = (A\dot\cup B, E_H)$.
\item Use the algorithm from Lemma~\ref{lem:derandomized} to find $\pi_k$ such that $|E_H \cap E_{G_{\pi_k}}| \geq |E_H| \cdot \frac{|E_G|}{n^2}$. \label{step:find-perm}
\end{enumerate}
\item Output $\pi_1, \dots, \pi_k$.
\end{enumerate}
It is obvious that the permutations are all side-preserving permutations and that the union of $E_{G_{\pi_i}}$ over $i\in [k]$ is equal to $E_{K_{n,n}}$. To see that $k\le \frac{2n^2\ln n}{|E_G|}+1$, observe that due to the guarantee of Lemma~\ref{lem:derandomized}, $|E_H|$ decreases by a multiplicative factor of (at most) $(1 - |E_G|/n^2) \leq e^{-|E_G|/n^2}$ for each permutation picked. Since the set $E_H$ remains non-empty after $k - 1$ permutations are picked, we have $e^{-(k - 1) \cdot |E_G|/n^2} \cdot n^2 \geq 1$, which implies that $k \leq 2n^2 \ln n/|E_G| + 1$ as desired. Finally, the bottleneck in the running time is Step~\ref{step:find-perm}; we execute this step $k$ times and each execution takes $O(n^4)$ time. Thus, the total running time is $O(nk) = O(n^6 \log n)$.
\end{proof}

\end{document}